\documentclass[a4paper,11pt]{article}

\usepackage{geometry}

\usepackage{geometry}

\usepackage{amssymb}
\usepackage[utf8]{inputenc}
\usepackage[english]{babel}
\usepackage{amsmath, amstext, amsfonts, mathrsfs,amsthm}
\usepackage{amsfonts}
\usepackage{amssymb}
\usepackage{dsfont}
\usepackage{mathrsfs}
\usepackage{graphicx}
\usepackage{epsfig}
\usepackage[authoryear]{natbib}
\usepackage{color}
\usepackage{bm}
\usepackage{verbatim}
\usepackage{stmaryrd}
\usepackage[colorlinks=false,linkcolor=darkblue]{hyperref}
\usepackage{bbm}			
\usepackage{enumerate} 
\usepackage{nicefrac}
\usepackage{xfrac}
\usepackage{units} 
\usepackage{mathtools}
\usepackage{pdfcomment}
\usepackage{url}
\usepackage{tikz}
\usepackage[labelfont={bf,sf},font={small},%
  labelsep=space]{caption}
\usepackage{acronym}
\usepackage{mathrsfs}
\usepackage{fancyhdr}
\usepackage[T1]{fontenc}
\usepackage{caption, booktabs}
\usepackage{textcomp}
\usepackage{multirow}
\usepackage[normalem]{ulem}

\usepackage[header]{appendix}
\usepackage{tikz}
\usetikzlibrary{positioning,arrows}
\usetikzlibrary{matrix}
\usepackage{hologo}
\usepackage{xcolor}
\usetikzlibrary{matrix,fit}
\tikzset{%
  mleftdelimiter/.style={inner ysep=0pt, inner xsep=1ex,left delimiter=\{,label={[label distance=3mm]left:#1}}
}

\usepackage[Bjarne]{fncychap} 

\usepackage{setspace}
\usepackage[squaren]{SIunits}

\usepackage{overpic}
\usepackage{subcaption}

\newcommand{\E}{\mathbb{E}}

\newcommand{\N}{\mathbb{N}}

\newcommand{\R}{\mathbb{R}}

\newcommand{\cM}{\mathcal{M}}

\newcommand{\cF}{\mathcal{F}}

\newcommand{\cA}{\mathcal{A}}

\newcommand{\eqd}{\stackrel{\mathrm{d}}=}

\newcommand{\1}{\mathds{1}}

\newcommand{\de}{\mathrm{\,d}}

\newcommand{\Ran}{\mathsf{Ran}}

\newcommand{\Var}{\mathrm{Var}}

\newcommand{\setcounterprefix}[2]{%
  \setcounter{#1}{0}%
  \expandafter\def\csname theH#1\endcsname{#2.\arabic{#1}}
  \expandafter\def\csname the#1\endcsname{#2.\arabic{#1}}%
}

\DeclareMathOperator{\v@r}{V@R}
\DeclareMathOperator{\avar}{AVaR}
\DeclareMathOperator{\av@r}{AV@R}

\newtheorem{theorem}{Theorem}[section]
\newtheorem{proposition}[theorem]{Proposition}

\newtheorem{corollary}[theorem]{Corollary}

\newtheorem{lemma}[theorem]{Lemma}
\newtheorem{example}[theorem]{Example}
\newtheorem{remark}[theorem]{Remark}

\makeatletter

\@addtoreset{equation}{section}
\makeatother

\numberwithin{figure}{section}%
\numberwithin{table}{section}

 \renewcommand\appendix{\par
   \setcounter{section}{0}%
   \setcounter{subsection}{0}%
   \setcounter{figure}{0}%
   \renewcommand\thesection{\Alph{section}}%
   \renewcommand\thefigure{\Alph{section}.\arabic{figure}}}

\usepackage[backgroundcolor=white,bordercolor=orange]{todonotes}

\usepackage{accents}

\title{Robust Bernoulli mixture models for credit portfolio risk}

\author{Jonathan Ansari$^{1}$ and Eva L\"{u}tkebohmert$^{2}$}

\begin{document}

\maketitle

\vspace{-6ex}
\begin{center}
\small\textit{$^{1}$Department of Mathematics, \\ University of Salzburg,
Hellbrunner Stra{ss}e 34, 5020 Salzburg, Austria.\\
$^{2}$Department of Quantitative Finance,\\ Institute for Economic Research, University of Freiburg,\\ Rempartstr. 16,
79098 Freiburg, Germany.}                                                                              \end{center}

\begin{abstract}
This paper presents comparison results and establishes risk bounds for credit portfolios within classes of Bernoulli mixture models, 
assuming conditionally independent defaults that are stochastically increasing with a common risk factor.
We provide simple and interpretable conditions for conditional default probabilities that imply a comparison of credit portfolio losses in convex order. In the case of threshold models, the ranking of portfolio losses is based on a pointwise comparison of the underlying copulas.
Our setting includes as special case the well-known Gaussian copula model but allows for general tail dependencies, which are crucial for modeling credit portfolio risks. Moreover, our results extend the classical parameterized models, such as the industry models CreditMetrics and KMV Portfolio Manager, to a robust setting where individual parameters or the copula modeling the dependence structure can be ambiguous. 
A simulation study and a real data example under model uncertainty offer evidence supporting the effectiveness of our approach.

\end{abstract}



\section{Introduction}

A fundamental problem in modeling credit portfolio risk is that defaults are rare events, resulting in very limited available data.
While the default probabilities of individual entities can generally be estimated with reasonable accuracy, e.g. using historical data or financial and macroeconomic indicators, and can therefore be treated as known, modeling the typically positive dependencies between default events poses a greater challenge. These dependencies, driven by common exposure to economic cycles, market shocks, and industry-specific risk factors, are complex to capture accurately due to the curse of dimensionality.
For this reason, factor models for credit portfolio risk have been introduced, which are economically interpretable, easier to estimate, have well understood properties, and are reliable under different scenarios.
Well known examples of industry credit risk models are CreditMetrics (cf. \cite{CreditMetrics}), KMV's PortfolioManager (cf. \cite{Kealhofer2001}), and CreditRisk$^+$ (cf. \cite{CreditRisk+}). The first two models belong to the class of Bernoulli mixture models (BMMs) where default is described by the event when the firm's latent asset return falls below a given threshold reflecting its liabilities. 
In contrast, the CreditRisk$^+$ model is an actuarial approach that approximates the Bernoulli default indicator with a Poisson random variable, resulting in a Poisson mixture model as an approximation of the BMM.
However, due to the scarcity of data, estimating the dependencies between defaults of different entities remains a challenging task, even in factor models. This leads to ambiguity in some parameter estimates such as the asset correlation coefficient. 
Additionally, asset returns often exhibit tail dependencies, which are typically underestimated by the widely used Gaussian factor models. Consequently, tail risks are likely to be underrepresented in simulated portfolio losses. 
These challenges highlight the need to study robustness in credit risk models, accounting for both tail dependencies and model uncertainty. \\

In this paper, we establish robustness results for stochastically increasing Bernoulli mixture models (siBMMs), where the conditional default probabilities of individual borrowers increase with a common risk factor, which may be either a latent random variable or an observable random variable, such as a market index.
Our main contribution, Theorem \ref{propcomBMMs}, is a comparison result for credit portfolio losses where we provide straightforward and interpretable conditions on the conditional default probability functions that lead to a convex ordering of the losses. More specifically, the loss variable of portfolio $A$ is smaller in convex order than that of portfolio $B$ if the default events in portfolio $B$ exhibit a stronger positive dependence on the (potentially portfolio-specific) common risk factor than those in portfolio~$A$. 
The proof relies on a general construction of siBMMs via copula-based threshold models. While most of the paper assumes that default probabilities increase with the common risk factor, Theorem \ref{thebmmsiBMM} presents a robustness result that holds even without this monotonicity assumption.
We then address model uncertainty by establishing lower and upper bounds for credit portfolio losses in convex order (Theorem \ref{mainthe}) and for credit portfolio risks in terms of law-invariant, convex risk measures (Corollary \ref{cormaithe}). 
Our robustness results enable a simple comparison of different models, allowing risk measures to be ranked across various model specifications. As an additional contribution of our paper, we provide a comparison of credit portfolio losses in convex order through a pointwise comparison of the copulas underlying the threshold models (Proposition \ref{lemconord}). 
We illustrate our results in two contexts: first, when individual parameters are ambiguous, and second, when there is uncertainty about the copula modeling the dependence between default events. 
In both cases, we calculate bounds for the Average Value-at-Risk (AVaR), a benchmark risk measure in the Basel III and Solvency II frameworks, using both simulated and real portfolios.

\bigskip

Our paper relates to a large literature on credit risk models. For a general introduction to credit risk modeling, we refer to \cite{Bluhm} among many others. 
Latent factor models have been studied in \citet{Koyluoglu-1998}, \cite{Crouhy}, and \citet{Gordy-2000}, who also point out the equivalence between threshold models and mixture models for the special case of CreditMetrics and CreditRisk\(^+\). 
\cite{FreyMcNeilNyfeler} relate factor models to copula models. More specifically, they show that the distribution of the number of defaults in a portfolio strongly depends on the underlying copula. We consider copula-based threshold models and link them to the class of siBMMs.

\cite{FreyMcNeilNyfeler} also point out that the Gaussian copula is unable to reflect tail dependence (compare also \cite{Embrechts-2015}, Section 11). 
\cite{DonnellyEmbrechts2010} highlight the inability of the Gaussian copula model to reflect default clustering and the associated tail dependence. 
\cite{CrookMoreira2011} empirically document that non-Gaussian copula families can better model dependence structures of credit portfolios. The impact of the choice of the copula model in stress testing applications has been empirically studied, for example, in \cite{KoziolSchellEckhardt2015}. 
Their results indicate that the Gaussian copula is suitable when analysing the effect of stress scenarios characterized by extreme losses and high correlations while heavy-tailed copulas are more appropriate in less extreme adverse scenarios.
Non-Gaussian copula models for the modelling of dependent defaults in intensity models have first been introduced in \cite{Schoenbucher2001}. 
We consider flexible dependence structures described by general stochastically increasing copulas, including as important examples the Gaussian and various Archimedean copula settings.

Our paper also relates to the literature on model uncertainty. \cite{FreyMcNeilNyfeler} document that portfolio credit risk models are subject to substantial model risk, with small variations in the dependence structure leading to significant changes in the loss distribution. This highlights the need for a more robust approach, which is even more important given the difficulty in reliably calibrating these models (compare also \cite{Frey-McNeil-2003}). 
Worst case bounds for credit portfolio risks with and without dependence information are studied for the Value-at-Risk (VaR) in \cite{Bernard-2018,Bernard-2018c}, for the Tail VaR in \cite{Bernard-2014}, for the Range VaR  in \cite{Bernard-2020}, \cite{Li-2018}, and for distortion risk measures in \cite{Pesenti-2024}. \cite{Embrechts-2013} provide a rearrangement algorithm to explicitly calculate sharp lower and upper bounds for portfolio VaR when the marginal distributions of individual risks are known and show that these bounds can be improved when higher order marginal information is available. 
More generally, risk bounds for sums of random variables 
are developed in \cite{Wang-2011} and \cite{Embrechts-2006}, when only the components' marginal distributions are known, and in \cite{Denuit-1999}, when additionally the joint distribution of the components is bounded by some known distribution.
These bounds are typically too wide to be practically relevant, which motivates to include additional dependence information and structural assumptions to narrow the bounds. For instance, \cite{Bernard-2017} consider as additional constraint that the variance of the portfolio sum is upper bounded. \cite{Bernard-2017b} assume a partially specified factor model, where the joint distributions of each component and the common risk factor are known but the conditional distributions are unspecified.
\cite{Lux-2019} improve risk bounds by considering an upper and a lower bound on the copula of the risk factors. \cite{Puccetti-2017} provide risk bounds for the VaR of sums of random variables, assuming the components' marginal distributions are known and subgroups of components are independent. \cite{Bignozzi-2015} study bounds for convex risk measures when marginal distributions are known and there is additional positive or negative dependence information. 
\cite{Bernard-Vanduffel-2015} assume dependence information under certain scenarios, leading to significantly tighter risk bounds.

We build on the aforementioned literature by, on the one hand, extending standard industry models to a robust setting that accommodates model uncertainty and tail dependence. On the other hand, our robustness results for BMMs may also be seen in the context of improved credit portfolio risk bounds under the structural assumption of conditional independence---an assumption that is both tractable and flexible, allowing for a broad spectrum of positive dependencies. 
In addition, our ordering results enable a ranking of portfolio losses in terms of their credit risk, which provides a transparent and objective basis for the assessment of portfolio risk levels.

The rest of the paper is organized as follows. Section \ref{intro} establishes our main result---a convex ordering of credit portfolio losses for siBMMs. Section \ref{seccompres} derives lower and upper bounds for credit portfolio losses. Section~\ref{secappl} illustrates our results in both simulated and real data examples. Appendix \ref{appendA} provides a representation of siBMMs via copula-based threshold models, while all proofs are deferred to Appendix \ref{appendB}.

\section{Setting and Main Result}\label{intro}

We consider a portfolio of $N$ loans to borrowers $n=1,\ldots,N$, and denote by $\pi_n$ the (unconditional) default probability and by $D_n \colon \Omega \to \{0,1\}$  the default indicator of borrower $n$, where \(D_n = 1\) indicates default and \(D_n = 0\) no default\footnote{ 
We generally assume that the underlying probability space \((\Omega,\cA,P)\) is atomless and, hence, admits a continuous random variable.
}.
Let \(Z\) be a real-valued random variable, which will serve as a common risk factor.
As default models, we consider a class of \emph{Bernoulli mixture models} (BMMs) \(\{(D_1,\ldots,D_N,Z)\}\,,\) where 
\begin{enumerate}[(I)]
    \item \label{ass1} \(\mathbb{P}(D_n=1) = 1-\mathbb{P}(D_n=0) = \pi_n\in [0,1]\) and
    \item \label{ass2} \(D_1,\ldots,D_N\) are conditionally independent given \(Z.\)
\end{enumerate}
Denote by \(p_{D_n}(z):=\mathbb{P}(D_n=1 \mid  Z=z)\) the conditional default probability function. If, additionally to \eqref{ass1} and \eqref{ass2},
\begin{enumerate}[(I)]\setcounter{enumi}{2}
    \item \label{ass3} \(p_{D_n}\) is an increasing function in \(z\,,\)
\end{enumerate}
we refer to \((D_1,\ldots,D_N,Z)\) as \emph{stochastically increasing Bernoulli mixture model} (siBMM), noting that a random variable \(X\) is \emph{stochastically increasing (SI)} in \(Y\) if, for all \(x\in \R\,,\) \(\mathbb{P}(X\geq x \mid Y=y)\) is increasing in \(y\).
The individual default probabilities \(\pi_1,\ldots,\pi_N\) can often be estimated and thus assumed to be known, so Assumption \eqref{ass1} is reasonable. Since default events are typically not independent, but rather positively dependent, the dependence structure between \(D_1,\ldots,D_N\) is crucial for modeling portfolio risks. 
We impose conditional independence in \eqref{ass2} as standard structural assumption for modeling simple dependencies, where the common risk factor \(Z\) may be an observable factor, e.g. a market index, or a latent random variable that follows an arbitrary distribution, e.g. a discrete or continuous distribution.
Assumption \eqref{ass3} incorporates positive dependence into our model, stating that
the default indicators \(D_n\) are stochastically increasing in \(Z.\)
This assumption is typically not imposed in the definition of a BMM, but is fulfilled by many models; see \cite[Section 11.2]{Embrechts-2015}. In practice, the risk factor \(Z\) may be interpreted as a market index which has a positive or negative influence on the default probabilities. For example, if the market deteriorates, default events may be more likely. In the latter case, \(Z\) is a decreasing transformation of a market index. 
The assumptions of a siBMM allow to model various cases of positive dependence for \((D_1,\ldots,D_n)\,,\) in particular, independence (where the function \(p_{D_n}\) is constant for all \(n\)) and, if \(F_Z\) is continuous, also comonotonocity\footnote{A bivariate random vector \((X,Y)\) is said to be \emph{comonotonic} if there exist increasing functions \(f\colon \R\to \R\) and \(g\colon \R\to \R\) and a random variable \(U\) such that \(X=f(U)\) and \(Y=g(U)\) almost surely. In this regard, comonotonicity models perfect positive dependence.} (where \(p_{D_n}\) takes the form \(p_{D_n}(z) = \1_{\{z\geq F_Z^{-1}(1-\pi_n)\}}\) for all \(n\)). The latter case means that default is a deterministic function of \(Z.\)

Let \(e_n\in (0,\infty) \) denote the (deterministic) exposure at default of borrower $n$ and \(\delta_n\in (0,1] \) its (percentage) loss given default, which is assumed to be a random variable that is independent of all other sources of randomness. 
Then the bank's total loss is given by
\begin{align}\label{equ L}
    L:=\sum_{n=1}^N e_n \delta_n D_n.
\end{align} 
\begin{figure}
    \centering
    \includegraphics[scale = 0.6, trim={0 30 0 40},clip]{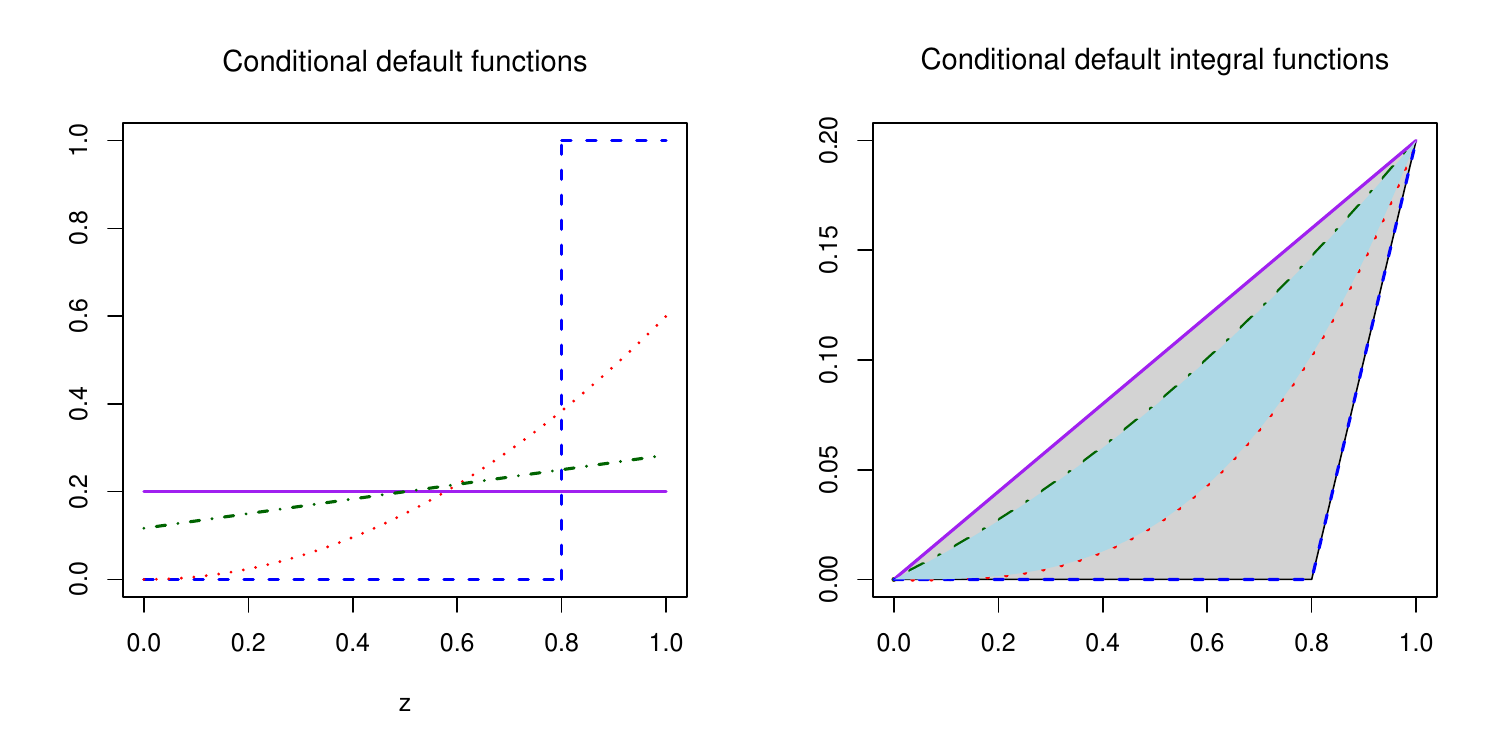}
    \caption{Left plot: Some examples of transformed conditional default probabilities \(p_{D_n}\circ F_Z^{-1} = \mathbb{P}(D_n =1 \mid Z = F_Z^{-1}(\cdot))\) which are by Assumption \eqref{ass3} increasing and satisfy the constraint \(\int_0^1 \mathbb{P}(D_n =1 \mid Z = F_Z^{-1}(t)) \de t = \pi_n\) with \(\pi_{n}\) chosen as \(0.2.\) The solid line models independence of \(D_n\) and \(Z\,,\) while the dashed line models perfect positive dependence of \(D_n\) on \(Z.\) Right plot: Default integral functions \(G_{D_n,Z}\) defined in \eqref{defcdpf} as the antiderivatives of the respective transformed conditional default probabilities from the left plot, noting that each \(G_{D_n,Z}\) is increasing convex and passes through the points \((0,0)\) and \((0,\pi_{n}).\) Under Assumption \eqref{ass3}, the independent and comonotonic dependence relations mark the extreme cases (solid and dashed) while, for all other dependencies, the default integral functions lie between these extremes.}
    \label{fig:transformed cond PDs}
\end{figure}
As a main contribution, we provide simple and interpretable conditions on the conditional default probabilities \(p_{D_n}\) to imply a comparison of credit portfolio losses in convex order. Recall that, for non-negative random variables \(X\) and \(Y,\) the convex order \(X\leq_{cx} Y\) is defined by \(\E \varphi(X) \leq \E \varphi(Y)\) for all convex functions \(\varphi\colon [0,\infty) \to \R;\) see \cite[Section 3.A]{Shaked-Shantikumar-2007} for an overview of the convex order.
More precisely, our conditions involve a pointwise comparison of the \emph{default integral functions}
\begin{align}\label{defcdpf}
    G_{D_n,Z}(s) :=  \int_0^s \mathbb{P}(D_n = 1 \mid Z = F_Z^{-1}(t)) \de t
    = \int_0^s p_{D_n}(F_Z^{-1}(t)) \de t, \quad s\in [0,1],
\end{align}
 associated with \(p_{D_n}\). Here, \(F_Z^{-1}\) denotes the generalized inverse of the distribution function of \(Z\,,\) i.e., \(F_Z^{-1}(t) := \inf \{ x \mid F_Z(x)\geq t\}\) for \(t\in (0,1).\)  Thus, $G_{D_n,Z}(s)$ gives the average conditional default probability over the lowest $s$-quantile of the risk factor $Z$.
It is obvious that \(G_{D_n,Z}\) is a continuous increasing function mapping from \((0,1)\) to \([0,\pi_n]\) such that \(G_{D_n,Z}(0) = 0\) and \(G_{D_n,Z}(1) = \pi_n.\) Further, under Assumption \eqref{ass3}, \(G_{D_n,Z}\) is convex and lies between the extreme cases of independence and comonotonicity as discussed before; see Figure \ref{fig:transformed cond PDs}. 
The following theorem is our main result.

\begin{theorem}[Comparison of losses in siBMMs]\label{propcomBMMs}~\\
Consider default models \((D_1,\ldots,D_N,Z)\) and \((D_1',\ldots,D_N',Z')\) that satisfy Assumptions \eqref{ass1}--\eqref{ass3} with unconditional default probabilities $(\pi_1,\ldots,\pi_N)$.
Let \(\delta_1,\ldots,\delta_N\) be independent non-negative random variables that are also independent of the default models and let \(e_1,\ldots,e_N\) be positive constants. 
     If \(G_{D_n,Z}(s) \geq G_{D_n',Z'}(s)\) for all \(n\) and \(s,\) then
    \begin{align}\label{eqpropcomBMMs}
        \sum_{n=1}^N e_n \delta_n D_n \leq_{cx} \sum_{n=1}^N e_n \delta_n D_n'.
    \end{align}
\end{theorem}

As a consequence of the above theorem, we establish in Section \ref{seccompres} robustness results for credit risk models and determine in Theorem \ref{mainthe} lower and upper bounds in convex order for classes of credit portfolios models.

\begin{remark}\label{remmaithe}
\begin{enumerate}[(a)]
\item \label{remmaithea} The pointwise comparison of default integral functions in Theorem \ref{propcomBMMs} allows a simple interpretation: Smaller values of \(G_{D_n,Z}\) indicate more positive dependence between \(D_n\) and \(Z.\) More precisely,
if \(Z\) has the same distribution as \(Z',\) then \(G_{D_n,Z}(s) \geq G_{D_n',Z'}(s)\) for all \(s\) is equivalent to \(\mathbb{P}(D_n >  x, Z>z) \leq \mathbb{P}(D_n' > x,Z'>z)\) for all $x,z$ or to \(\mathbb{P}(D_n=1\mid Z>z) \leq \mathbb{P}(D_n' =1\mid Z'>z)\) for all \(z,\) i.e., default is more likely in the model \((D_n',Z')\) when the risk factor exceeds fixed values. The above relation is well studied in the literature and known as \emph{upper orthant order} or \emph{more positive quadrant dependence order}; see e.g. \cite[Chapter 6]{Rueschendorf-2013} or \cite{Joe-1997}. 
\item The proof of Theorem \ref{propcomBMMs} is based on a representation of siBBMs through threshold models (see Appendix \ref{appendA}) and on a supermodular comparison result for factor models (see \cite[Corollary 4.1]{Ansari-Rueschendorf-2023}). Since the supermodular order implies the upper orthant order, see \cite[Figure 3.1]{Mueller-Stoyan-2002}, we also obtain under the assumptions of Theorem \ref{propcomBMMs} that joint default probabilities are ordered, i.e., \(\mathbb{P}(D_n = 1, n\in M) \leq \mathbb{P}(D_n'=1, n\in M)\) for all \(M\subseteq \{1,\ldots,N\}.\)
    \item The default integral function \(G_{D_n,Z}\) is linear whenever \(D_n\) and \(Z\) are independent. In particular, under Assumption \eqref{ass3}, \(G_{D_n,Z}(s) = s \pi_n\) for some \(s\in (0,1)\) is equivalent to independence of \(D_n\) and \(Z.\) Further, if \(F_Z\) is continuous, then \(G_{D_n,Z}(s) = (s-1+\pi_{n})\1_{\{s \geq 1-\pi_{n}\}}\) whenever \(D_n\) and \(Z\) are comonotonic. In particular, in this case,
    \(G_{D_n,Z}(s) = 0\) for \(s = 1-\pi_n\)
    is equivalent to comonotonicity of \(D_n\) and \(Z.\)
\end{enumerate}
\end{remark}

The focus of this paper is on siBMMs, i.e. under Assumption \eqref{ass3} that default probabilities are stochastically increasing in the common risk factor. 
However, if this assumption is dropped, it remains possible to upper bound the losses of a credit portfolio in convex order by leveraging the concept of increasing rearrangements. This leads to the following theorem, which asserts that any BMM can be dominated by a siBMM.
Therefore, let \(f\colon (0,1)\to \R\) be an integrable function. The \emph{increasing rearrangement} of \(f\) is the (essentially with respect to the Lebesgue measure uniquely determined) increasing function \(f^*\) which satisfies \(\int_x^1 f(t) \de t \leq \int_x^1 f^*(t) \de t\) for all \(x\in (0,1)\) with equality for \(x=0,\) see \cite[Theorem 3.13]{Rueschendorf-2013}. We denote by \(\eqd\) equality in distribution.

\begin{theorem}[Domination of BMMs by siBMMs]\label{thebmmsiBMM}~\\
    For fixed \(\pi_n\in (0,1),\) \(n\in \{1,\ldots,N\},\) let \((D_1,\ldots,D_N,Z)\) be a BMM and \((D_1',\ldots,D_N',Z')\) be a siBMM with \(Z\eqd Z'.\)
    Let \(\delta_1,\ldots,\delta_N\) be independent non-negative random variables that are also independent of the default models and let \(e_1,\ldots,e_N\) be positive constants. 
    If \(u \mapsto p_{D_n'} \circ F_{Z'}^{-1}(u)\) is the increasing rearrangement of \(u \mapsto p_{D_n} \circ F_{Z}^{-1}(u),\) then 
    \begin{align}\label{eqthebmmsiBMM}
        \sum_{n=1}^N e_n \delta_n D_n \leq_{cx} \sum_{n=1}^N e_n \delta_n D_n'.
    \end{align}
\end{theorem}

\begin{remark}
    Under Assumption \eqref{ass3}, the function \(u \mapsto p_{D_n} \circ F_{Z}^{-1}(u)\) in Theorem \ref{thebmmsiBMM} is increasing and thus coincides with its increasing rearrangement \(u \mapsto p_{D_n'} \circ F_{Z'}^{-1}(u).\) Further, due to Theorem \ref{propcomBMMs}, the upper bound in \eqref{eqthebmmsiBMM} is in convex order between the independent loss \(L^\perp\) and the comonotonic loss \(L^c,\) which are given by
    \begin{align}\label{defstandbou}
    L^\perp := \sum_{n=1}^N e_n \delta_n D_n^\perp \quad \text{and} \quad L^c := \sum_{i=1}^N e_n \delta_n D_n^c,
\end{align}
where \(D_1^\perp,\ldots,D_N^\perp\) are independent and \(D_1^c,\ldots,D_N^c\) are comonotonic Bernoulli-distributed random variables with \(\mathbb{P}(D_n^\perp = 1) = \mathbb{P}(D_n^c = 1) = \pi_n\) for all \(n\in \{1,\ldots,N\}.\)
\end{remark}

\section[Robustness Results]{Robustness Results for Credit Risk Models}\label{seccompres}

In this section, we 
determine lower and upper risk bounds for classes of credit risk models under parameter or model uncertainty. Additionally, we examine simple conditions under which these improved bounds are attained.

For a vector of borrower specific default probabilities \(\pi = (\pi_n)_{1\leq i \leq N},\) 
denote by 
\begin{align}
    \cM^\pi = \{(D_1,\ldots,D_N,Z) ~\text{satisfying conditions \eqref{ass1}-\eqref{ass3}}
    \}
\end{align} 
the class of siBMMs.
Denote by \(\cF_{icx}\) the class of increasing convex functions mapping from \([0,1]\) to \(\R.\) For \(\alpha\in [0,1],\) we consider the subclass
\begin{align}
    \cF_{icx}^\alpha := \{f \in \cF_{icx} \mid   f(0)= 0,  f(1) = \alpha, ||f||_L\leq 1\}
\end{align}
of increasing convex functions 
mapping from \((0,0)\) to \((1,\alpha)\) with Lipschitz constant \(||\cdot||_L\) being not larger than \(1.\) For an arbitrary family \(f=(f_t)_{t\in I}\subset \cF_{icx}^\alpha,\) we denote by
\begin{align*}
   \overline{\mathop{co}}(f) &:= \inf\{ g \mid g~\text{convex}, g(y)\geq \sup_{t\in I} f_t(y) ~\forall y\},\\
 \underline{\mathop{co}}(f)&:= \sup\{ g \mid g~\text{convex}, g(y)\leq \inf_{t\in I} f_t(y) ~\forall y\},
\end{align*}
the least convex majorant and greatest convex minorant. It is straightforward to verify that also \(\underline{\mathop{co}}(f)\) and \(\overline{\mathop{co}}(f)\) are in \(\cF_{icx}^\alpha.\)
We make use of the function class \(\cF_{icx}^\alpha\) to describe default integral functions as defined in \eqref{defcdpf}. 
It can easily be seen that any default integral function associated with an siBMM is in the class \(\cF_{icx}^\alpha\) for \(\alpha = \pi_n.\) Vice versa, for every function \(g\in \cF_{icx}^\alpha,\) there exists a siBMM with \(\pi_n = \alpha\) having default integral function \(G_{D_n,Z} = g.\)

The following theorem addresses classes of siBMMs and determines improved lower and upper bounds in convex order for the associated credit portfolio losses. These improvements are relative to the standard bounds for sums of positively dependent random variables within the same Fr\'{e}chet class, given by the independent loss \(L^\perp\) and the comonotonic loss \(L^c\).

\begin{theorem}[Bounds in convex order for classes of siBMMs]\label{mainthe}~\\
    For fixed default probabilities \(\pi = (\pi_1,\ldots,\pi_N),\) let \(\{(D_{1}^x,\ldots,D_{N}^x,Z^x)\}_{x\in I}\)
    be a family of siBMMs with \((D_{1}^x,\ldots,D_{N}^x,Z^x)\in \cM^\pi\), with default integral functions \(G_{D^x_{n},Z^x}\), and with losses \(L^x := \sum_{n=1}^N e_n\delta_n D_{n}^x\) for all \(x\in I.\)  
    Consider siBMMs \((\underline{D}_1,\ldots,\underline{D}_N,\underline{Z})\) and \((\overline{D}_1,\ldots,\overline{D}_N,\overline{Z})\) in \(\cM^\pi\) specified by the 
    default integral functions
    \begin{align}\label{eqdif}
    \underline{G}_{n} &:= \overline{\mathop{co}}\left((G_{D_{n}^x,Z^x})_{x\in I}\right) \quad \text{and} \quad \overline{G}_{n} := \underline{\mathop{co}}\left((G_{D_{n}^x,Z^x})_{x\in I}\right),
    \end{align}
    respectively. Then, we have
    \begin{align}
        L^{\perp}\leq_{cx} \underline{L} \leq_{cx} L^x \leq_{cx} \overline{L} \leq_{cx} L^c \quad \text{for all } x\in I,
    \end{align}
    where \(\underline{L} := \sum_{n=1}^N e_n \delta_n \underline{D}_n\) and \(\overline{L} := \sum_{n=1}^N e_n \delta_n \overline{D}_n.\)
\end{theorem}

\begin{remark}\label{rembco}
    \begin{enumerate}[(a)]
    \item Theorem \ref{mainthe} imposes no assumption on the mixing variables \(Z^x\). Instead, the bounds are derived solely based on the least convex majorant and greatest convex minorant of the family of default integral functions, ensuring broad applicability and minimal structural constraints. The functions \(\underline{G}_n\) and \(\overline{G}_n\) in \eqref{eqdif} again define siBMMs, which yield the bounds \(\underline{L}\) and \(\overline{L}\) for \(L^x\) in convex order. Note that \(\underline{G}_n \geq \overline{G}_n\) pointwise, in accordance with Remark \ref{remmaithe}\eqref{remmaithea} where pointwise smaller default integral functions indicate stronger positive dependencies. 
    \item \label{rembco1} The lower bound \(\underline{L}\) (upper bound \(\overline{L}\)) in Theorem \ref{mainthe} is sharp if there exists a sequence \((x_m)_{m\in \N},\) \(x_m\in I\) for all \(m,\) of siBMMs such that \(G_{D_{n}^{x_m},Z^{x_m}}(s) \to \underline{G}_{n}(s)\) (\(G_{D_{n}^{x_m},Z^{x_m}}(s) \to \overline{G}_{n}(s)\)) for all \(s\) as \(m\to \infty.\) In this case, the derivatives of the convex default integral functions converge pointwise (outside a Lebesgue-null set) so that the assertion follows from the dominated convergence theorem. 
    The conditions are satisfied, for instance, in a parametric subfamily of Gaussian or Clayton copula-based threshold models, where the underlying parameters lie within specified intervals; see also Examples \ref{exBMM} and \ref{excomthr}.
        \item If in a BMM \((D_1,\ldots,D_N,Z)\) the risk factor \(Z\) is independent of the default variables \(D_1,\ldots,D_N,\) then we have \(\mathbb{P}(D_n=1\mid Z=F_Z^{-1}(t)) =  \pi_n\) and thus \(G_{n,Z}(t) = t \pi_n\) for all \(t\in [0,1]\) and for all \(n.\) 
        Hence, the lower bound \(\underline{L}\) improves the loss \(L^\perp\) of the independence model in convex order if and only if there exists \(n\in \{1,\ldots,N\}\) and \(s\in (0,1)\) such that \(\underline{G}_{n,Z}(s)\ne s \pi_n .\)  
        Similarly, if \(D_1,\ldots,D_N,Z\) are comonotonic and \(F_Z\) is continuous, then \(\mathbb{P}(D_n=1\mid Z=F_Z^{-1}(t)) =  \1_{\{t \geq 1- \pi_n\}}\) and thus \(G_{n,Z}(t) = (t-1+\pi_n) \1_{\{t\geq 1-\pi_n\}}\) for all \(t\) and for all \(n.\) Hence,
        the upper bound \(\overline{L}\) improves the comonotonic bound \(L^c\) in convex order if and only if there exists \(n\in \{1\ldots,N\}\) and \(t\in (0,1)\) such that \(\overline{G}_{n,Z}(t) \ne (t-1+\pi_n) \1_{\{t\geq 1-\pi_n\}}.\)
        \item For a family \(f = \{f_t\}_{t\in I}\) of functions in \(\cF_{icx}^\alpha,\) denote by
\begin{align}\label{genF}
    \mathfrak{c}(f) := \left\{ h \in \cF_{icx}^\alpha \mid \underline{\mathop{co}}(f)(s) \leq h(s) \leq \overline{\mathop{co}}(f)(s) ~ \text{for all } s\in (0,1)\right\}
\end{align}
the class of increasing convex functions generated by \(f.\) 
The functions \(g_1(s) = \alpha s\) and \(g_2(s) = (s-1+\alpha)\1_{\{s\geq 1-\alpha\}}\) generate the class \(\cF_{icx}^\alpha,\) i.e., \(\mathfrak{c}(\{g_1,g_2\}) = \cF_{icx}^{\alpha}.\) 
The right-hand plot in Figure \ref{fig:transformed cond PDs} illustrates for \(\alpha = 0.2\) the classes \(\mathfrak{c}(\{g_1,g_2\})\) and \(\mathfrak{c}(\{h_1,h_2\})\) as the area between the respective functions, where \(h_1(s) = 1/12 s^2 + 7/60 s\) and \(h_2(s) = 0.2 s^3.\) 
    \end{enumerate}
\end{remark}

To quantify portfolio losses in terms of risk measures, we consider convex, law-invariant risk measures defined on the space of integrable random variables on \((\Omega,\cA,P).\) It is well known that convex law-invariant risk measures are consistent with the convex order if the underlying probability space is atomless; see \cite{Bauerle-2006,Rueschendorf-2006,Schachermayer-2006}. For an overview of risk measures, we refer to \cite[Chapter 4]{Follmer-Schied-2016}.
The following result is a direct consequence of Theorem \ref{mainthe}.

\begin{corollary}[Lower and upper credit portfolio risk bounds in siBMMs]\label{cormaithe}~\\
    Let \(\psi\) be a law-invariant, convex risk measure on \((\Omega,\cA,P)\).
    Under the assumptions of Theorem \ref{mainthe}, we have for all \(x\in I\) that
    \begin{align}
        \psi(L^\perp) \leq \psi(\underline{L}) \leq \psi(L^x) \leq \psi(\overline{L}) \leq \psi(L^c).
    \end{align}
\end{corollary}
An important example of a law-invariant, convex risk measure is the AVaR at level \(\alpha\in (0,1),\) defined for an integrable random variable \(S\)  by
\begin{align}
    \avar_\alpha(S) := \frac 1 \alpha \int_{1-\alpha}^1 F_S^{-1}(t) \de t.
\end{align}
In the Basel III and Solvency II frameworks, the AVaR is a benchmark risk measure in the internal models for market risk. It was introduced as a replacement for VaR in order to better capture tail risks in the revised approach to risk measurement implemented in \cite{Basel2013}. Note that the VaR is not a convex risk measure, but is still used as the benchmark under the Internal Ratings Based (IRB) approach of the Basel regulatory framework for the calculation of minimum capital requirements for credit portfolios. In the following section, we apply our robustness results to data and study improved bounds for the AVaR for families of siBMMs under parameter uncertainty and for classes of siBMMs under model uncertainty.

\section{Applications}\label{secappl}

We illustrate our findings on the robustness of credit portfolio models with a simulation study and a real data example.


\subsection{Simulation study}

We perform a simulation study considering four different classes of threshold models.
In threshold models, default is defined as the event where the value of a financial asset drops below a specified threshold, typically corresponding to the company's liability value; see Appendix \ref{appendA}. 
As a standard threshold model, we consider the Gaussian copula model which, however, has the drawback that it does not admit tail dependencies. Given the strong tail dependencies commonly observed in  
financial assets, copula models that account for such dependencies are generally preferable. For this reason, we specifically examine the Clayton copula family, which belongs to the class of Archimedean copulas, along with its survival copulas. 
As a fourth class, we examine a hybrid approach that combines Gaussian and Clayton copula models. We then compare the bounds derived under model uncertainty for this hybrid model with the individual bounds obtained for each copula model under parameter uncertainty.

Examples of the Gaussian copula model include the classical industry models of KMV, see \cite{Kealhofer2001}, as well as the (default-mode) CreditMetrics model, see \cite{CreditMetrics}, which also allows for an extension to a rating transition model. The Clayton copula model has been studied, for instance, in \cite[Example 8.22]{Embrechts-2015}.

Our analysis is based on an application of Corollary \ref{cormaithe} in order to determine uncertainty intervals for the AVaR of credit portfolio losses in the respective classes of siBBMs. 
To this end, we compare the associated default integral functions \(G_{D_n,Z}\) pointwise. 
By Proposition \ref{lemconBMM} and Lemma \ref{correpBMM}, every siBMM \((D_1,\ldots,D_N,Z)\) can be identified with a copula-based threshold model \(D_n = \1_{\{X_n \leq c_n\}}\), where the copula of \(C_{X_n,Y}\) is SI, \(n\in \{1,\ldots,N\}\); see Appendix \ref{secmain} for all details. If the distribution function of the common risk factor \(Z\) is continuous, we obtain from \eqref{survcop} the representation
\begin{align}\label{repdefintfun}
    G_{D_n,Z}(t) = \pi_n - C_{X_n,Y}(\pi_n,1-t)
\end{align}
for the respective threshold model.

In the following example, we consider copula-based threshold models, where 
the dependence structure of \((X_n,Y)\) is described by a Gaussian, a Clayton, and a survival Clayton copula, respectively. In each case, the underlying copula parameter is subject to uncertainty.

\begin{example}[siBMMs under parameter uncertainty]\label{exthreshmod}~\\
We consider three siBMMs implied by copula-based threshold models (due to Proposition \ref{lemconBMM}) under parameter uncertainty and apply Corollary \ref{cormaithe}
to obtain the minimal and maximal credit portfolio risks for the respective family of models.  The probability of default is set to \(\pi_n = 0.02\) for all \(n\in \{1,\ldots,N\}.\)
\begin{figure}[t]
    \centering
    \includegraphics[scale = 0.5, trim={0 30 0 15},clip]{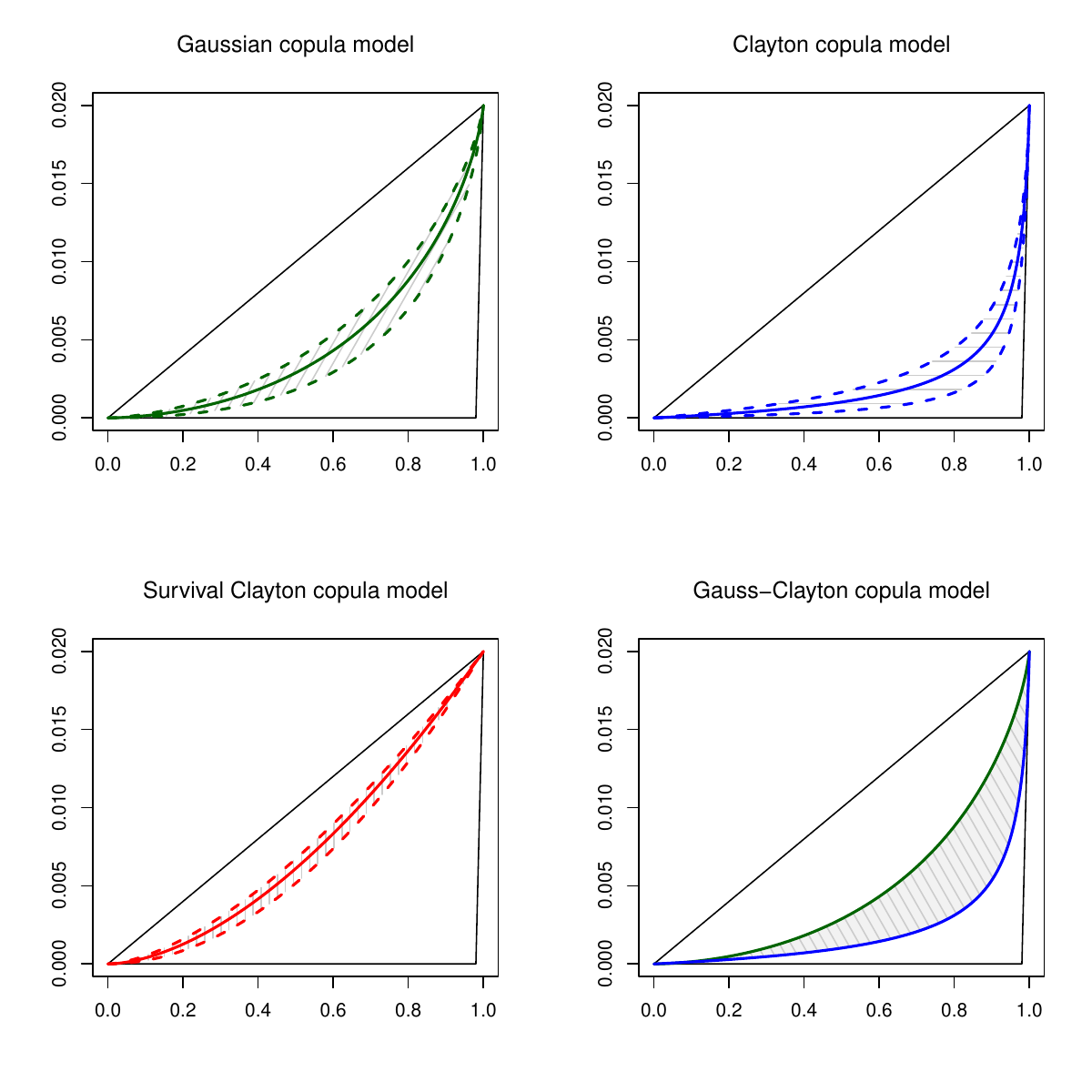}
    \caption{
    Default integral functions \(G_{D_n,Z}\) of the threshold model classes considered in Examples \ref{exthreshmod} and \ref{exthreshmod2} under parameter and model uncertainty: Gaussian copula model (green), Clayton copula model (blue), and survival Clayton copula model (red). The solid lines are obtained from asset correlations estimated by the IRB formula, where the dashed lines \(\overline{G}_n^{Ga}\), \(\underline{G}_n^{Ga}\), \(\overline{G}_n^{Cl}\), \(\underline{G}_n^{Cl}\), \(\overline{G}_n^{sCl}\), \(\underline{G}_n^{sCl}\) correspond to the asset correlation bounds  of the IRB formula \eqref{equ IRB asset correlation}.
    }
    \label{fig:enter-label}
\end{figure}
\begin{enumerate}[(a)]
    \item \label{exthreshmoda} Gaussian copula model: For \(n\in \{1,\ldots,N\},\) we consider a family \(\{D_n^{\rho_n} = \1_{\{X_n\leq c_n\}}\}_{\rho_n \in I_n}\) of threshold models, where \(X_n\) follows the Merton model \eqref{Vasicek} with unspecified asset correlation \(\rho_n\in I_n :=[0.12,0.24]\). The bounds of the interval \(I_n\) for \(\rho_n\) are chosen as the standard bounds \(0.12\) and \(0.24\) in the IRB approach (see \cite{IRB2005,Basel2011}), where the asset correlation is determined as
\begin{align}\label{equ IRB asset correlation}
    \rho_n=0.12\cdot \frac{1-e^{-50\cdot \pi_n}}{1-e^{-50}}+0.24\cdot \left(1-\frac{1-e^{-50\cdot \pi_n}}{1-e^{-50}}\right),
\end{align}
which is a function of the borrower's unconditional default probability $\pi_n$.
    The joint distribution of \(X_n\) and the latent risk factor \(Y\) is bivariate normal with correlation \(\sqrt{\rho_n}.\) 
    Hence, the associated copula is the Gaussian copula with parameter \(\sqrt{\rho_n}.\) The threshold is given as \(c_n = \Phi^{-1}(\pi_n);\) see \eqref{eqthreshold}. 
    According to \eqref{equ IRB asset correlation} and Example \ref{exBMM} \eqref{exBMMa}, the implied family of siBMMs \(\{(D_1^{\rho_1},\ldots,D_N^{\rho_N},Z)\}_{\rho_1,\ldots,\rho_N}\) is specified by the default integral functions
    \begin{align}\label{defintfunIRBGauss}
    \begin{split}
        G_{D_n^{\rho_n},Z}(s) 
        = \pi_n - C_{\sqrt{\rho_n}}^{Ga}(\pi_n,1-s) 
        = \pi_n - \Phi_{\sqrt{\rho_n}}(\Phi^{-1}(\pi_n),\Phi^{-1}(1-s)) ,
    \end{split}
    \end{align}
    where \(\Phi_\rho\) is the standardized bivariate normal cdf with correlation \(\rho\) and \(\Phi\) is the univariate standard normal cdf. 
    Since the Gaussian copula family \((C^{Ga}_\rho)_{\rho\in [-1,1]}\) is increasing in \(\rho\) with respect to the pointwise order of copulas and since Gaussian copulas with non-negative parameter are SI, \(G_{D_n^{\rho_n},Z}(s)\) is decreasing in \(\rho_n\) in the pointwise order. Hence, the default integral bounds due to Theorem \ref{mainthe} are given by
    \begin{align*}
        \overline{G}^{Ga}_n(s) &:= \underline{\mathop{co}}\left((G_{D_n^{\rho_n},Z})_{\rho_n\in I})\right)(s) = G_{D_n^{0.24},Z}(s) = \pi_n - \Phi_{\sqrt{0.24}}(\Phi^{-1}(\pi_n),\Phi^{-1}(1-s)),
        \\ 
        \underline{G}^{Ga}_n(s) &:= \overline{\mathop{co}}\left((G_{D_n^{\rho_n},Z})_{\rho_n\in I})\right)(s) = G_{D_n^{0.12},Z}(s) = \pi_n - \Phi_{\sqrt{0.12}}(\Phi^{-1}(\pi_n),\Phi^{-1}(1-s)),
    \end{align*}
    see also Remark \ref{rembco}\eqref{rembco1}.
    The upper left plot in Figure \ref{fig:enter-label} illustrates the range of the default integral functions in the Gaussian threshold model under parameter uncertainty, where, for all \(n\in \{1,\ldots,N\},\) the lower (upper) dashed graph corresponds to the default integral bound \(\overline{G}^{Ga}_n\) (\(\underline{G}^{Ga}_n\)). 
    The solid line indicates the default integral function in the IRB approach for the asset correlation \(\rho_n \approx 0.165\) obtained from formula \eqref{equ IRB asset correlation} for our choice \(\pi_n = 0.02.\)
    \item \label{exthreshmodb} Clayton copula model: Analogous to \eqref{exthreshmoda}, we consider for each \(n\) a family \(\{D_n^{\theta_n} = \1_{\{X_n\leq c_n\}}\}_{\theta_n \in I}\) of threshold models, where now \(X_n\) depends on the latent factor \(Y\)  via the Clayton copula model \eqref{Claytoncopmod} with unspecified dependence parameter \(\theta_n\in I_n :=[\underline{\theta},\overline{\theta}].\) For reasons of comparability, the interval bounds \(\underline{\theta} := 4\pi \arcsin(0.12)/(1-2\pi \arcsin(0.12))\) and \(\overline{\theta} := 4\pi \arcsin(0.24)/(1-2\pi \arcsin(0.24))\) are chosen as the parameters of the Clayton copula having the same Kendall's tau value as the Gaussian copula with the parameter bounds from the IRB formula as discussed in part \eqref{exthreshmoda}. Since copulas have uniform marginals, the threshold level is given by \(c_n = \pi_n.\) The default integral functions of the implied family \(\{(D_1^{\theta_1},\ldots,D_N^{\theta_N},Z)\}_{\theta_1,\ldots,\theta_N}\) of siBMMs are given by 
    \begin{align}\label{defintfunCl}
        G_{D_n^{\theta_n},Z}(s) = \pi_n - C_{\theta_n}^{Cl}(\pi_n, 1-s),
    \end{align}
    using \eqref{repdefintfun}.
    Since the Clayton copula family is increasing in \(\theta\) with respect to the pointwise order and since Clayton copulas are SI for a non-negative parameter, we obtain the following sharp lower and upper default integral bounds: 
    \begin{align*}
        \overline{G}^{Cl}_n(s) &:= \underline{\mathop{co}}\big((G_{D_n^{\theta_n},Z})_{\rho_n\in I})\big)(s) 
        = G_{D_n^{\overline{\theta}},Z}(s) 
        = \pi_n - \left(\pi_n^{-\overline{\theta}}+(1-s)^{-\overline{\theta}}-1\right)^{-1 / \overline{\theta}},\\ 
        \underline{G}^{Cl}_n(s) &:= \overline{\mathop{co}}\left((G_{D_n^{\rho_n},Z})_{\rho_n\in I})\right)(s) = G_{D_n^{\underline{\theta}},Z}(s) 
        = \pi_n - \big(\pi_n^{-\underline{\theta}}+(1-s)^{-\underline{\theta}}-1\big)^{-1 / \underline{\theta}}.
    \end{align*}
    The upper right plot in Figure \ref{fig:enter-label} illustrates the range of default integral functions in the Clayton threshold model under parameter uncertainty, where, for all \(n\in \{1,\ldots,N\},\) the lower (upper) dashed graph corresponds to the default integral bound \(\overline{G}^{Cl}_n\) (\(\underline{G}^{Cl}_n\)). 
    The solid line indicates the default integral function for \(\theta_n \approx 0.723\) for which the Clayton copula has the same Kendall's tau value as the Gaussian copula in \eqref{exthreshmoda} obtained from the IRB formula.
    \item \label{exthreshmodc} Survival Clayton copula threshold model: We consider the setting in \eqref{exthreshmodb}, where we now replace the Clayton copula family in \eqref{defintfunCl} by the survival Clayton copula family. Noting that Kendall's tau value and the SI property are invariant under taking the survival copula transformation \eqref{survcop}, we obtain the lower and upper default integral bounds 
    \begin{align*}
        \overline{G}^{sCl}_n(s) &:= \underline{\mathop{co}}\big((G_{D_n^{\theta_n},Z})_{\rho_n\in I})\big)(s) = G_{D_n^{\overline{\theta}},Z}(s) 
        = s - \left((1-\pi_n)^{-\overline{\theta}}+s^{-\overline{\theta}}-1\right)^{-1 / \overline{\theta}},\\ 
        \underline{G}^{sCl}_n(s) &:= \overline{\mathop{co}}\big((G_{D_n^{\rho_n},Z})_{\rho_n\in I})\big)(s) = G_{D_n^{\underline{\theta}},Z}(s) 
        = s - \left((1-\pi_n)^{-\underline{\theta}}+s^{-\underline{\theta}}-1\right)^{-1 / \underline{\theta}}.
    \end{align*}
    The lower left plot in Figure \ref{fig:enter-label} illustrates the range of default integral functions in the survival Clayton threshold model under parameter uncertainty. We observe that the range of default integral functions in this setting significantly differs from the Clayton copula setting because we now have upper instead of lower tail dependencies.
\end{enumerate}
\end{example}

Tables \ref{tab:avar1} and \ref{tab:avar2} highlight the significant impact of accounting for lower tail dependencies in a credit portfolio comprising \(N = 1.000\) loans, each with a homogeneous default probability of \(\pi_n = 0.02\) for all \(n\in \{1,\ldots,N\}.\) As underlying convex risk measure, we choose the $\avar_{\alpha}$ with levels \(\alpha=95\%\) and \(\alpha=99\%.\) In the first scenario, we consider a deterministic loss given default of \(\delta_n \equiv 0.1\) (see Table \ref{tab:avar1}); in the second scenario, we choose a beta-distributed loss given default \(\delta_n\) with expectation \(\E[\delta_n] = 0.1\) and volatility \(\sqrt{\Var(\delta_n)} = 0.15\) (see Table \ref{tab:avar2}).
In the settings of the default models analysed in Example \ref{exthreshmod}, our simulations demonstrate that accurately capturing tail dependencies is more critical than accounting for parameter uncertainty within the individual models. This is evidenced by the non-overlapping uncertainty intervals observed in the Gaussian, Clayton, and survival Clayton copula settings. 
The lower bound of the uncertainty interval in the survival Clayton copula setting is only marginally higher than the risk 
\(\avar_\alpha^\perp := \avar_\alpha(L^\perp)\)
of the credit portfolio with independent default events. This can be attributed to the fact that  the lower tail dependencies in the survival Clayton copula are nearly independent (as evident from the lower left corner of the sample generated from the survival Clayton copula in Figure \ref{fig:sampleCopulaModels}).
In contrast, the Clayton copula setting captures lower tails with a notably strong positive dependence, effectively accounting for the pronounced co-movements in extreme events (see the lower left corner of the sample from the Clayton copula in Figure \ref{fig:sampleCopulaModels}). As a consequence, the upper bound for the Clayton copula is about two third of the worst-case comonotonic bound \(\avar_\alpha^c := \avar_\alpha(L^c).\)
The tail dependencies in the Gaussian setting lie somewhere in between, suggesting that portfolio risks modeled using the Gaussian model might initially appear to provide a realistic approach. However, as previously noted, Gaussian copulas typically underestimate tail dependencies in financial applications. Consequently, the Clayton copula setting may offer a more accurate reflection of the associated risks.

To find a balance between the Gaussian and the Clayton copula settings, we introduce a class of distributions whose tail dependencies lie between those of the specified Gaussian and Clayton copula. More precisely, we consider a class siBMMs, where the default integral functions are generated by a combination of Gaussian and Clayton copula threshold models.

\begin{example}[siBMMs under model uncertainty]\label{exthreshmod2}~\\
We consider a class of siBMMs specified by the default probabilities \(\pi_n = 0.02\) for all \(n\) and by the conditional default probability functions \(p_{D_n}\) associated 
with the class \(\mathfrak{c}(\{G_n^{Ga},G_n^{Cl}\})\) (defined in \eqref{genF}) of default integral functions. Here \(G_n^{Ga}\) is the default integral function in \eqref{defintfunIRBGauss} associated with the Gaussian copula with parameter \(\rho_n \approx 0.165\) implied by the IRB formula. 
The function \(G_n^{Cl}\) is the default integral function in \eqref{defintfunCl} associated with the Clayton copula with parameter \(\theta_n \approx 0.723\, ,\) which accounts for tail dependencies.
The lower right-hand plot in Figure \ref{fig:enter-label} illustrates the class \(\mathfrak{c}(\{G_n^{Ga},G_n^{Cl}\})\) of default integral functions. We observe that the derivative of the function \(G_n^{Cl}\) (blue graph) is larger than that of \(G_n^{Ga}\) (green plot) for \(s\) close to \(1.\) This implies that the probability of default increases for large values of the common risk factor \(Z\), leading to a higher likelihood of simultaneous defaults.
\end{example}

\begin{table}[thb]
    \centering
    \scalebox{0.80}{
    \begin{tabular}{c||c|c||c|c||c|c||c|c||c|c}
        \toprule
        \multirow{2}{*}{\vspace{0.5em} Scenario 1} & \multicolumn{2}{c||}{Gaussian} & \multicolumn{2}{c||}{Clayton} & \multicolumn{2}{c||}{Surv. Clayton} & \multicolumn{2}{c||}{Gauss-Clayton} & \multicolumn{2}{c}{} \\
        
        & \(\underline{\avar}_\alpha\) & \(\overline{\avar}_\alpha\) & \(\underline{\avar}_\alpha\) & \(\overline{\avar}_\alpha\) & \(\underline{\avar}_\alpha\) & \(\overline{\avar}_\alpha\) & \(\underline{\avar}_\alpha\) & \(\overline{\avar}_\alpha\) & \(\avar_\alpha^\perp\) & \(\avar_\alpha^c\) \\
        \midrule
        \(\alpha=95\) & 0.80\% & 1.21\% & 2.02\% & 2.83\% & 0.37\% & 0.44\% & 0.95\% & 2.37\% & 0.30\% & 4.02\% \\
        \(\alpha=99\) & 1.17\% & 2.00\% & 4.45\% & 6.56\% & 0.42\% & 0.49\% & 1.47\% & 5.35\% & 0.33\% & 10.0\% \\
        \bottomrule
    \end{tabular}}
    \caption{Lower and upper bounds for the AVaR of the percentage losses in the robust credit risk models under parameter uncertainty and model uncertainty as specified in Examples \ref{exthreshmod} and \ref{exthreshmod2}. The number of loans is chosen as \(N = 1.000,\) the loss given default is \(\delta_n \equiv 0.1\) and the exposure at default as fraction of the total amount is \(e_n = 1/N\) for all \(n\in \{1,\ldots,N\}.\) The numbers are obtained by Monte Carlo simulations with a sample size of \(1.000.000.\)}
    \label{tab:avar1}
\end{table}

The risk bounds for the AVaR in the hybrid 'Gauss-Clayton' setting specified in Example \ref{exthreshmod2} are also illustrated in Tables \ref{tab:avar1} and \ref{tab:avar2} for the two different scenarios of loss given default as discussed previously. The risk bounds capture model uncertainty for the class of default integral functions lying between the Gaussian and Clayton case. Recall that Figure \ref{fig:enter-label} illustrates the classes of default integral functions considered in Examples \ref{exthreshmod} and \ref{exthreshmod2}.

\begin{table}[thb]
    \centering
    \scalebox{0.80}{
    \begin{tabular}{c||c|c||c|c||c|c||c|c||c|c}
        \toprule
        \multirow{2}{*}{\vspace{0.5em} Scenario 2} & \multicolumn{2}{c||}{Gaussian} & \multicolumn{2}{c||}{Clayton} & \multicolumn{2}{c||}{Surv. Clayton} & \multicolumn{2}{c||}{Gauss-Clayton} & \multicolumn{2}{c}{} \\
        
        & \(\underline{\avar}_\alpha\) & \(\overline{\avar}_\alpha\) & \(\underline{\avar}_\alpha\) & \(\overline{\avar}_\alpha\) & \(\underline{\avar}_\alpha\) & \(\overline{\avar}_\alpha\) & \(\underline{\avar}_\alpha\) & \(\overline{\avar}_\alpha\) & \(\avar_\alpha^\perp\) & \(\avar_\alpha^c\) \\
        \midrule
        \(\alpha=95\) & 0.83\% & 1.24\% & 2.03\% & 2.84\% & 0.46\% & 0.51\% & 0.99\% & 2.38\% & 0.39\% & 4.02\% \\
        \(\alpha=99\) & 1.22\% & 2.02\% & 4.46\% & 6.58\% & 0.54\% & 0.61\% & 1.50\% & 5.36\% & 0.46\% & 10.4\% \\
        \bottomrule
    \end{tabular}}
    \caption{Lower and upper bound for the AVaR of the percentage losses in the robust credit risk models under parameter uncertainty and model uncertainty as specified in Examples \ref{exthreshmod} and \ref{exthreshmod2}. The number of loans is chosen as \(N = 1.000,\) the loss given default \(\delta_n\) is \(beta(0.3,2.7)\)-distributed, and the exposure at default as fraction of the total amount is \(e_n = 1/N\) for all \(n\in \{1,\ldots,N\}.\) The numbers are obtained by Monte Carlo simulations with a sample size of \(1.000.000.\)}
    \label{tab:avar2}
\end{table}

\subsection{Real data example}\label{secreadatexa}

Next, we compare credit portfolio risk bounds derived from different copula models for realistic sovereign loan portfolios that we construct from publicly available data. Therefore, we extract information from the financial statements of the Inter American Development Bank (IDB) (see \cite{IDB2022}) 
and combine this with data on sovereign ratings provided by different rating agencies. Table \ref{tab IDB portfolio} summarizes the portfolio's exposures (in million USD) and corresponding borrower ratings. 
Further, we derive the sovereign default probabilities \(\pi_n\) from the rating transitions as published in \cite[Table 35]{S&P2022}. 
Again, we consider a deterministic loss given default rate of $\delta_n \equiv 10\%$ or a beta distributed loss given default \(\delta_n\) with mean \(10\%\) and volatility equal to \(15\%\). The relatively low loss given default reflects a high preferred creditor status (compare \cite{S&P2018}) whereby sovereigns tend to repay loans to supranational institutions before other creditors.
\begin{table}[tb]
    \centering
    \scalebox{0.65}{
    \begin{tabular}{r|l|r|l||r|r||r|r||r|r|}
    \toprule
    \multicolumn{10}{l}{\textbf{Panel B: IDB}}\\
    n & Country	& Amount & Rating & \(\pi_n\) [\%] & \(e_n\) [\%] & \(\underline{\rho}_n\) [\%] & \(\overline{\rho}_n\) [\%] & \(\underline{\theta}_n\) & \(\overline{\theta}_n\) \\
    \hline
    1  & Argentina           & 15548  & CCC-  &84.80  & 14.33  & 6.00  & 16.00  & 0.37  & 0.71 \\
    2  & Bahamas             & 725    & B+    & 1.46  & 0.67   & 13.71 & 23.71  & 0.64  & 0.96 \\
    3  & Barbados            & 629    & B-    & 7.59  & 0.58   & 6.36  & 16.36  & 0.39  & 0.72 \\
    4  & Belize              & 157    & B-    & 7.59  & 0.15   & 6.36  & 16.36  & 0.39  & 0.72 \\
    5  & Bolivia             & 4012   & B-    & 7.59  & 3.70   & 6.36  & 16.36  & 0.39  & 0.72 \\
    6  & Brazil              & 15185  & BB-   & 0.90  & 13.99  & 16.20 & 26.20  & 0.72  & 1.04 \\
    7  & Chile               & 2290   & A     & 0.01  & 2.11   & 21.92 & 31.92  & 0.90  & 1.24 \\
    8  & Colombia            & 11150  & BB+   & 0.18  & 10.27  & 20.62 & 30.62  & 0.86  & 1.19 \\
    9  & Costa Rica          & 2454   & B+    & 1.46  & 2.26   & 13.71 & 23.71  & 0.64  & 0.96 \\
    10 & Dominican Republic  & 3971   & BB    & 0.40  & 3.66   & 19.10 & 29.10  & 0.81  & 1.14 \\
    11 & Ecuador             & 7541   & B-    & 7.59  & 6.95   & 6.36  & 16.36  & 0.39  & 0.72 \\
    12 & El Salvador         & 2310   & CCC+  &18.20  & 2.13   & 6.00  & 16.00  & 0.37  & 0.71 \\
    13 & Guatemala           & 1913   & BB    & 0.40  & 1.76   & 19.10 & 29.10  & 0.81  & 1.14 \\
    14 & Guyana              & 787    & BB-   & 0.90  & 0.73   & 16.20 & 26.20  & 0.72  & 1.04 \\
    15 & Haiti               & 0      & CCC   &50.40  & 0.00   & 6.00  & 16.00  & 0.37  & 0.71 \\
    16 & Honduras            & 3068   & BB-   & 0.90  & 2.83   & 16.20 & 26.20  & 0.72  & 1.04 \\
    17 & Jamaica             & 1692   & BB-   & 0.90  & 1.56   & 16.20 & 26.20  & 0.72  & 1.04 \\
    18 & Mexico              & 15417  & BBB   & 0.06  & 14.21  & 21.53 & 31.53  & 0.89  & 1.22 \\
    19 & Nicaragua           & 2316   & B     & 2.38  & 2.13   & 10.87 & 20.87  & 0.54  & 0.87 \\
    20 & Panama              & 4357   & BBB   & 0.06  & 4.01   & 21.53 & 31.53  & 0.89  & 1.22 \\
    21 & Paraguay            & 3071   & BB    & 0.40  & 2.83   & 19.10 & 29.10  & 0.81  & 1.14 \\
    22 & Peru                & 3158   & BBB   & 0.06  & 2.91   & 21.53 & 31.53  & 0.89  & 1.22 \\
    23 & Suriname            & 656    & CCC-  &84.80  & 0.60   & 6.00  & 16.00  & 0.37  & 0.71 \\
    24 & Trinidad and Tobago & 732    & BBB-  & 0.11  & 0.68   & 21.14 & 31.14  & 0.87  & 1.21 \\
    25 & Uruguay             & 3370   & BBB+  & 0.04  & 3.11   & 21.68 & 31.68  & 0.89  & 1.23 \\
    26 & Venezuela           & 2011   & CC    &100.00 & 1.85   & 6.00  & 16.00  & 0.37  & 0.71 \\
    \hline
    & Total amount & 108520 & & & 100 & & & & \\
    \bottomrule
    \end{tabular}}
    \caption{Borrowing countries, credit rating \protect\footnote{according to S\&P, Fitch or Moody's ratings, \url{https://tradingeconomics.com/}, or obtained by regressing OECD ratings on S\&P ratings if ratings were not available} and amount outstanding (in million USD) as of 2022 (Source: \cite[p.83]{IDB2022}).
    \(\pi_n\) denotes the default probability and \(e_n\) is the relative amount of the exposure in the total portfolio. The last two pairs of columns refer the asset correlation uncertainty intervals \(I_n = [\underline{\rho}_n,\overline{\rho}_n]\) in the Gaussian model (Example \ref{exthreshmod}\eqref{exthreshmoda}) and the parameter uncertainty interval \(I_n = [\underline{\theta}_n,\overline{\theta}_n]\) in the Clayton models (Examples \ref{exthreshmod}\eqref{exthreshmodb} and \eqref{exthreshmodc}), where we set \(\underline{\rho}_n := \rho_n - 5\%\) and \(\overline{\rho}_n := \rho_n + 5\%\) with the asset correlation \(\rho_n\) obtained from the IRB formula \eqref{equ IRB asset correlation} depending on \(\pi_n\) and the bounds chosen as \(0.11\) and \(0.27.\) 
    }
    \label{tab IDB portfolio}
\end{table}
\footnotetext{According to S\&P, Fitch or Moody's ratings, \url{https://tradingeconomics.com/}, or obtained by regressing OECD ratings on S\&P ratings if ratings were not available.}

\begin{table}[htb]
    \centering
    \scalebox{0.80}{
    \begin{tabular}{c||c|c||c|c||c|c||c|c||c|c}
        \toprule
        \multirow{2}{*}{\vspace{0.5em} Scenario 1} & \multicolumn{2}{c||}{Gaussian} & \multicolumn{2}{c||}{Clayton} & \multicolumn{2}{c||}{Surv. Clayton} & \multicolumn{2}{c||}{Gauss-Clayton} & \multicolumn{2}{c}{} \\
        
        & \(\underline{\avar}_\alpha\) & \(\overline{\avar}_\alpha\) & \(\underline{\avar}_\alpha\) & \(\overline{\avar}_\alpha\) & \(\underline{\avar}_\alpha\) & \(\overline{\avar}_\alpha\) & \(\underline{\avar}_\alpha\) & \(\overline{\avar}_\alpha\) & \(\avar_\alpha^\perp\) & \(\avar_\alpha^c\) \\
        \midrule
        \(\alpha=95\) & 2.72\% & 2.83\% & 2.96\% & 3.22\% & 2.67\% & 2.70\% & 2.77\% & 3.10\% & 2.64\% & 3.68\% \\
        \(\alpha=99\) & 3.32\% & 3.51\% & 4.27\% & 4.91\% & 3.21\% & 3.25\% & 3.41\% & 4.63\% & 3.18\% & 5.89\% \\
        \bottomrule
    \end{tabular}}
    \caption{Lower and upper bounds for the AVaR of the percentage losses for the real-data example discussed in Section \ref{secreadatexa}. The underlying robust credit risk models incorporate parameter uncertainty and model uncertainty in a Gaussian, Clayton and survival Clayton copula setting as in Examples \ref{exthreshmod} and \ref{exthreshmod2}; see the caption of Table \ref{tab IDB portfolio} for the respective parameters.
     The loss given default in this scenario is \(\delta_n \equiv 0.1\) and the numbers are obtained by Monte Carlo simulations with a sample size of \(1.000.000.\)}
    \label{tab:avar3}
\end{table}

\begin{table}[htb]
    \centering
    \scalebox{0.80}{
    \begin{tabular}{c||c|c||c|c||c|c||c|c||c|c}
        \toprule
        \multirow{2}{*}{\vspace{0.5em} Scenario 2} & \multicolumn{2}{c||}{Gaussian} & \multicolumn{2}{c||}{Clayton} & \multicolumn{2}{c||}{Surv. Clayton} & \multicolumn{2}{c||}{Gauss-Clayton} & \multicolumn{2}{c}{} \\
        
        & \(\underline{\avar}_\alpha\) & \(\overline{\avar}_\alpha\) & \(\underline{\avar}_\alpha\) & \(\overline{\avar}_\alpha\) & \(\underline{\avar}_\alpha\) & \(\overline{\avar}_\alpha\) & \(\underline{\avar}_\alpha\) & \(\overline{\avar}_\alpha\) & \(\avar_\alpha^\perp\) & \(\avar_\alpha^c\) \\
        \midrule
        \(\alpha=95\) & 8.44\% & 8.46\% & 8.48\% & 8.53\% & 8.44\% & 8.45\% & 8.45\% & 8.51\% & 8.44\% & 8.63\% \\
        \(\alpha=99\) & 11.19\% & 11.22\% & 11.27\% &11.36\% & 11.16\% & 11.17\% & 11.21\% & 11.32\% & 11.14\% & 11.55\% \\
        \bottomrule
    \end{tabular}}
    \caption{
    Lower and upper bounds for the AVaR of the percentage losses for the real-data example discussed in Section \ref{secreadatexa}. The underlying robust credit risk models incorporate parameter uncertainty and model uncertainty in a Gaussian, Clayton and survival Clayton copula setting as in Examples \ref{exthreshmod} and \ref{exthreshmod2}; see the caption of Table \ref{tab IDB portfolio} for the respective parameters.
     The loss given default in this scenario is \(\delta_n \sim beta(0.3,2.7)\)-distributed and the numbers are obtained by Monte Carlo simulations with a sample size of \(1.000.000.\)
     }
    \label{tab:avar4}
\end{table}

In the benchmark specification, we choose the asset correlation $\rho_n$ for each borrower according to the IRB approach (\ref{equ IRB asset correlation}) as a function of the unconditional default probabilities $\pi_n$ for each borrower
with the 12\% and 24\% bounds replaced by 11\% and 27\%. This choice is based on estimates of the asset correlation parameter $\rho_n$ from equity index data (cf. \cite{RiskControl2023} or \cite{LuetkebohmertSesterShen2023}). To achieve this, we collect country-specific equity index data calculated by MSCI and a regional equity index for emerging markets in Latin America (MSCI EM Latin America) that is used as a common factor. We then calculate normalized log returns $X_n$ for each country's equity index series and afterwards regress the series $X_n$ on the regional index. The resulting regression coefficient corresponds to the asset correlation $\sqrt{\rho_n}$ in our single factor model. Note that equity indices are not available for every borrowing country in the portfolio. In fact, we were only able to collect data for four countries: Argentina, Brazil, Colombia and Peru. Our estimates for the asset correlation for the four Latin American countries in our data set range between 11\% and 27\% with an average asset correlation of 19\%. Thus, we adjusted the bounds in the IRB asset correlation formula accordingly. To allow for parameter uncertainty, we set lower and upper bounds for borrower-specific correlations as $\underline{\rho}_n=\rho_n-5\%$ and $\overline{\rho}_n=\rho_n+5\%,$ respectively. The resulting values are reported in Table \ref{tab IDB portfolio}.

While in the case of a homogeneous and rather large credit risk portfolio (Tables \ref{tab:avar1} and \ref{tab:avar2}), the dependencies between the default events are crucial, we find that in the rather small (in terms of number of borrowers) and very heterogeneous credit portfolio of IDB, the AVaR bounds are much more comparable (Tables \ref{tab:avar3} and \ref{tab:avar4}). The reason is that for such a small and concentrated (w.r.t. single borrowers) portfolio, systematic risk reflected by common dependence on the risk factor is less relevant whereas idiosyncratic risk plays a much more important role, compare \cite{LuetkebohmertSesterShen2023} and \cite{LuetkebohmertSester2024}. Moreover, we observe that the distribution of the losses given default has a strong influence on the total portfolio risk which is due to the fact that tail events result from a combination of bad realizations of the risk factor and high values for the loss given default variable. Thus, AVaR increases with the variance of loss given default.


\section*{Conclusion}

In this paper, we established a convex ordering result for credit portfolio losses in stochastically increasing BMMs based on transparent and interpretable assumptions on the conditional default probabilities.
For threshold models, we provided a related ordering result based on a pointwise comparison of bivariate copulas.
As a consequence, we obtained several robustness results for credit portfolio losses. Without assuming monotonicity of the conditional default probabilities in the common factor, we derived an upper bound on portfolio losses that improves the standard comonotonic bound. 
A detailed simulation study and a real-data example illustrate the relevance and effectiveness of our approach. 
Our methods are especially applicable to standard industry models of credit portfolio risk, such as CreditMetrics or KMV Portfolio Manager, and allow for the incorporation of tail dependencies and model uncertainty into these models. In this way, our paper provides tractable and practically relevant methods for the assessment of credit portfolio risk under mild structural assumptions.

\section*{Acknowledgements}
The authors thank two anonymous referees for their valuable comments, which helped improve the manuscript.
The first author gratefully acknowledges the support of the Austrian Science Fund (FWF)
projects {10.55776/PAT1669224} \emph{SORT: Stochastic Orders for Functional Dependence}
{P 36155-N} \emph{ReDim: Quantifying Dependence via Dimension Reduction}
and the support of the WISS 2025 project 'IDA-lab Salzburg' (20204-WISS/225/197-2019 and 20102-F1901166-KZP).

\section*{Declarations of Interest}
The authors have no relevant financial or non-financial interests to declare.

\section*{Data availability}

Data is publicly available via \cite[p.83]{IDB2022} and \cite[Table~35]{S&P2022}.


\fancyhf{}
\fancyhead[RO]{REFERENCES}
\fancyhead[LO]{Robust Bernoulli mixture models}
\fancyfoot[C]{\thepage}



\newpage

\fancyhf{}
\fancyhead[RO]{APPENDIX}
\fancyhead[LO]{Robust Bernoulli mixture models}
\fancyfoot[C]{\thepage}


\appendices

\section*{Appendix}


\section{Threshold Models}\label{appendA}

We represent siBBMs in terms of copula-based threshold models and provide a comparison result for threshold models that we use for the proof of Theorem \ref{propcomBMMs}.

\subsection[Copula-Based Construction of siBMMs]{Copula-Based Construction of siBMMs via Threshold Models}\label{secmain}

It is well known that every BMM can be represented as a threshold model, i.e., the default indicator variable \(D_n\) is described by 
\begin{align}\label{defthreshmod}
    D_n = \1_{\{X_n\leq c_n\}}, \quad X_n = f_n(Y,\varepsilon_n),\quad n\in \{1,\ldots,N\},
\end{align}
for some default threshold value $c_n$, some measurable function $f_n \colon \R^2 \to \R$, and for independent idiosyncratic risk factors \(\varepsilon_1,\ldots,\varepsilon_N\) that are also independent from the random variable \(Y\) which models the common risk factor; see, e.g., \cite[Section 8.4.4]{Embrechts-2015}. 
If the distribution function of \(X_n\) is continuous, then, under Assumption \eqref{ass1}, it follows that 
\begin{align}\label{eqthreshold}
    c_n = F_{X_n}^{-1}(\pi_n).
\end{align} 
The dependence structure of \(X_n\) and \(Y\) is transferred to the dependence structure of \(D_n\) and \(Z,\) where we generally assume that \(Z\) is a strictly decreasing transformation of \(Y.\)

 A standard example of a threshold model is the single-factor Merton model (cf. \cite{Merton1974,Vasicek2002,CreditMetrics}), where
\begin{align}\label{Vasicek}
    X_n=\sqrt{\rho_n}\, Y+\sqrt{1-\rho_n}\, \epsilon_n
\end{align}
for i.i.d. standard normally distributed random variables \(Y,\varepsilon_1,\ldots,\varepsilon_N.\) 
In this case, \((X_n,Y)\) follows a bivariate normal distribution with zero mean, variance \(1,\) and  factor loading \(\sqrt{\rho_n},\) describing the sensitivity of borrower $n$ to the common factor $Y$. The asset return correlation between two distinct borrowers $n$ and $m$ is then given by $\sqrt{\rho_n\rho_m}$. 
The dependence structure between \(X_n\) and \(Y\) is described by a Gaussian copula with correlation \(\sqrt{\rho_n}.\) 
It can easily be seen that \(X_n\) in \eqref{Vasicek} is stochastically increasing in \(Y\), and thus the implied BMM satisfies Assumption \eqref{ass3}; see also Example \ref{exBMM}\eqref{exBMMa} 
below. 
However, in financial applications, the Gaussian copula is often not a good choice as it underestimates tail dependencies. This motivates to study robustness results for BMMs with general dependence structures.

In the sequel, we derive a copula-based characterization of siBMMs in terms of threshold models of the type
\eqref{defthreshmod}.
The following lemma shows that every threshold model \((D_n)_n\) 
yields a siBMM if the function \(f_n\) in \eqref{defthreshmod} satisfies some monotonicity condition. 

\begin{lemma}[Construction of siBMMs through threshold models]\label{lemBMM}~\\
Let \(X_n = f_n(Y,\varepsilon_n),\) \(n\in \{1,\ldots,N\},\) for some measurable function \(f_n\) that is increasing in its first component and for \(Y,\varepsilon_1,\ldots,\varepsilon_N\) being independent. 
Consider the default indicator random variables \((D_1,\ldots,D_N)\)  defined by \(D_n := \1_{\{X_n \leq c_n\}}\) for some threshold values $c_n$ and define \(Z:=g(Y)\) for some strictly decreasing function \(g\). Then the following statements hold:
\begin{enumerate}[(i)]
    \item \((D_1,\ldots,D_N)\) satisfies conditions \eqref{ass2} and \eqref{ass3} in the definition of an siBMM.
    \item If additionally \(Y,\varepsilon_1,\ldots,\varepsilon_N\) have a continuous distribution function and if, for all \(n\in \{1,\ldots,N\},\) \(f_n\) is continuous, then also condition \eqref{ass1} is satisfied for \(c_n = F_{X_n}^{-1}(\pi_n),\) i.e., \((D_1,\ldots,D_N,Z)\) is a siBMM with \(\mathbb{P}(D_n=1) = \pi_n\) for all \(n.\)
\end{enumerate}
\end{lemma}


We focus on the dependence structure of \((X_n,Y),\) which we generally model with copulas. The aim will be to investigate the relation between the copula \(C\) of \((X_n,Y)\) and the function \(f_n\) in Lemma \ref{lemBMM}. To this end, recall that
a \emph{bivariate copula} is a distribution function on \([0,1]^2\) with uniform univariate marginals. According to Sklar's theorem (see e.g. \cite{Nelsen-2006}), every bivariate distribution function \(F\) can be decomposed into its univariate marginal distribution functions \(F_1,F_2\) and a bivariate copula \(C\) such that
\begin{align}\label{eqSklar}
    F(x_1,x_2) = C(F_1(x_1),F_2(x_2))\quad \text{for all } (x_1,x_2)\in [0,1]^2.
\end{align}
The copula \(C\) is uniquely determined on the Cartesian product \(\Ran(F_1)\times \Ran(F_2),\) where \(\Ran(F_i)\) denotes the range of \(F_i.\) Further, for every bivariate copula \(C\) and all univariate distribution functions \(F_1\) and \(F_2,\) the right side of \eqref{eqSklar} defines a bivariate distribution function. Two important copulas are the \emph{product copula} \(\Pi(u,v):= uv,\) which models independence, and the \emph{upper Fr\'{e}chet copula} \(M(u,v):=\min\{u,v\},\) which models comonotonicity.
We refer to \cite{Nelsen-2006} and \cite{Durante-2016} for an overview of copulas.


Now, let \(U\) and \(V\) be real-valued random variables that are uniform on \((0,1).\) Then the joint distribution function of the bivariate random vector \((U,V)\) coincides with its copula \(C = C_{U,V}.\) 
For fixed \(u\in [0,1],\)
denote by 
\[ C(u|v):= \partial_2 C(u,v) := d C(u,v)/dv\] the partial derivative of \(C\) with respect to its second component, which exists for Lebesgue-almost all \(v;\) see e.g. \cite[Theorem 1.6.1]{Durante-2016}. Then the conditional distribution function of \(U|V=v\) can be represented, for fixed \(u\in [0,1],\) by
\begin{align}\label{repcondcdf}
    F_{U|V=v}(u) = \mathbb{P}(U\leq u \mid V=v) = C(u|v)
\end{align}
for Lebesgue-almost all \(v\in [0,1],\)
where \(C(u|v)\) is increasing but not necessarily right-continuous in \(u,\) see e.g. \cite[Theorem 3.4.4]{Durante-2016}. Denote by 
\begin{align}\label{defhfuninv}
    C^{-1}(t|v):=\inf \{u\in [0,1] \mid C(u|v)\geq t\}
\end{align}
the generalized inverse of the function \(u\mapsto C(u|v).\) 
For the construction of threshold models that imply BMMs satisfying monotonicity Assumption \eqref{ass3}, we make use of the notion of stochastically increasing (SI) copulas, that are copulas of random vectors \((U,V)\) such that \(U\) is SI in \(V.\)
As a consequence of \eqref{repcondcdf}, a bivariate copula \(C\) is SI if and only if \(C(u,v)\) is concave in \(v\) for all \(u.\)
Many well known (sub-)families of copulas are SI such as the Gaussian copula with non-negative correlation parameter, the Clayton copula or the Gumbel-Hougaard copula; see \cite{Ansari-Rockel-2023}
 for an overview of various SI copula families from classes of Archimedean, extreme-value, and elliptical copulas.

Due to the following result, we can construct the function \(f_n\) with the desired monotonicity condition in Lemma \ref{lemBMM} through a transformation of an SI copula that describes the dependence structure of \((X_n,Y).\) 
Denote by \(U(0,1)\) the uniform distribution on \((0,1)\).

\begin{proposition}[Copula-based construction of siBMMs through threshold models]\label{lemconBMM}~\\
Let \(C_1,\ldots,C_N\) be SI copulas and let \(Y,\varepsilon_1,\ldots,\varepsilon_n\sim U(0,1)\) be i.i.d. random variables.
For \(f_n(v,t):= C_n^{-1}(t|v),\) with the notation defined in \eqref{defhfuninv}, consider \(X_n := f_n(Y,\varepsilon_n).\) Then the following statements hold true:
\begin{enumerate}[(i)]
    \item \label{lemconBMM1} \(X_n \sim U(0,1),\)
    \item \label{lemconBMM2} \(f_n\) is componentwise increasing,
    \item \label{lemconBMM3} \(F_{X_n,Y} = C_{X_n,Y} = C_n.\)
\end{enumerate}
Further, for \(c_n := \pi_n\) and \(Z:=g(Y)\) with \(g\) strictly decreasing,
\((D_1,\ldots,D_N,Z)\) with \(D_n := \1_{\{X_n \leq c_n\}}\) for all \(n\) defines a siBMM and hence satisfies the conditions \eqref{ass1}--\eqref{ass3}.
\end{proposition}

\begin{remark}\label{remSIC}
\begin{enumerate}[(a)]
    \item A general construction method for BMMs through threshold models is given in \cite[Lemma 11.10]{Embrechts-2015}.  Our Proposition \ref{lemconBMM} is a copula-based version under the additional SI assumption.
    \item \label{remSIC2} 
    A copula of the bivariate default vector \((D_n,Z)\), constructed via the threshold model in Proposition \ref{lemconBMM}, is given by the survival copula of \(C_{X_n,Y}\), i.e.,
    \begin{align}\label{survcop}
        C_{D_n,Z}(u,v) = \hat{C}_{X_n,Y}(u,v):= u + v - 1 + C_{X_n,Y}(1-u,1-v)
    \end{align}
    for \((u,v)\in [0,1]^2.\) This follows from the fact that both \(D_n\) and \(Z_n\) are decreasing transformations of \(X_n\) and \(Y,\) respectively; see e.g. \cite[Theorem 2.4.4]{Nelsen-2006}.\footnote{The statement in this reference can be extended to random variables with discontinuous distribution function and to decreasing but not necessarily strictly decreasing transformations. 
    Hence, for any copula of two random variables, its survival copula is a copula of the decreasing transformations.}
    It can easily be verified that the survival copula \(\hat{C}\) is SI if and only if \(C\) is SI.
\end{enumerate}
\end{remark}

The following result is a reverse of Proposition \ref{lemconBMM}. It states that any siBMM can be written as a copula-based threshold model with SI copulas as in Proposition \ref{lemconBMM}. 
In other words, Proposition \ref{lemconBMM} yields a general copula-based construction method of siBMMs.

\begin{lemma}[Representation of siBMMs through threshold models]~\label{correpBMM}\\
    For any BMM \((D_1,\ldots,D_N,Z)\) satisfying conditions \eqref{ass1}--\eqref{ass3}, there exist random variables \(X_1,\ldots,X_N,Y\sim U(0,1)\) 
    such that 
    \begin{enumerate}[(i)]
        \item \(X_1,\ldots,X_N\) are conditionally independent given \(Y,\)
        \item \(C_{X_n,Y}\) is SI, and
        \item \(D_n = \1_{\{X_n\leq \pi_n\}}\).
    \end{enumerate}  
    Further, a copula of \((D_n,Z)\) is given by the survival copula \(\hat{C}_{X_n,Y}\) which is also SI.
\end{lemma}

In the following example, we discuss some well known SI copulas that can be used for the construction of siBMMs due to Proposition \ref{lemconBMM}.

\begin{example}[Copulas for siBMMs]\label{exBMM}
For modeling credit portfolio risks through threshold models of the form \eqref{defthreshmod}, lower tail dependencies between \(X_n\) and the common risk factor \(Y\) are crucial.
Gaussian copulas, which are often used as a standard model, suffer from weak tail dependencies. Stronger lower tail dependencies between \(X_n\) and \(Y\) can be modeled, e.g., with Clayton copulas. The dependence structure of a Gaussian, a Clayton and a survival Clayton copula having Kendall's tau value  \(\tau\approx 0.266\) is illustrated in Figure \ref{fig:sampleCopulaModels}. 
    \begin{enumerate}[(a)]
        \item \label{exBMMa} The Gaussian copula \(C_{\rho}^{Ga}(u,v) = \Phi_{\rho}(\Phi^{-1}(u),\Phi^{-1}(v))\) 
        with correlation parameter \(\rho\in [-1,1]\) is SI whenever \(\rho\geq 0\) because
            \(C_{\rho}^{Ga}(u|v) = \Phi((\Phi^{-1}(u) - \rho \, \Phi^{-1}(v))/\sqrt{1-\rho^2})\)
        is decreasing in \(v\) for all \(u.\) 
        For a Gaussian copula model \(F_{X_n,Y} = C_\rho^{Ga},\) we have by Proposition \ref{lemconBMM} the representation
        \begin{align}
            X_n = f_n(Y,\varepsilon_n) = (C_{\rho}^{Ga})^{-1}(\varepsilon_n|Y) = \Phi\left(\rho \, \Phi^{-1}(Y) + \sqrt{1-\rho^2}\, \Phi^{-1}(\varepsilon)\right),
        \end{align}
        which coincides with the single-factor Merton model in \eqref{Vasicek} for \(\rho=\sqrt\rho_n\) up to a quantile transformation of \(X_n\), \(Y\), and \(\varepsilon_n.\)        
        Since the bivariate normal distribution is radially symmetric, see \cite[Example 2.15]{Nelsen-2006}, it follows from \eqref{survcop} and \cite[Theorem 2.7.3]{Nelsen-2006} that \(C_{D_n,Z} = \hat{C}_{\rho}^{Ga} = C_{\rho}^{Ga}.\)

\begin{figure}
    \centering
    \includegraphics[scale = 0.54, trim={0 00 0 00},clip]{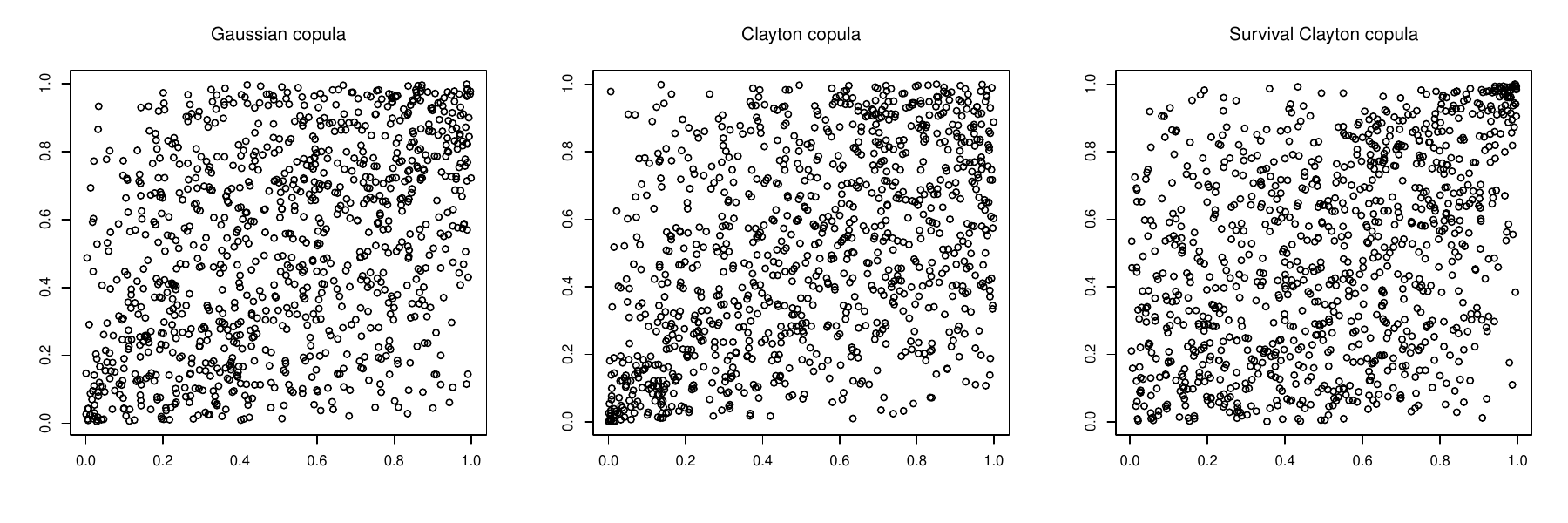}
    \caption{Samples from a Gaussian copula (left), a Clayton copula (mid) and a survival Clayton copula (right), where the Gaussian copula parameter is chosen as the square root \(\sqrt{\rho_n} \approx 0.405\) of the asset correlation obtained from the IRB formula \eqref{equ IRB asset correlation} for \(\pi_n = 0.02.\) The parameter \(\theta\) for the Clayton and survival Clayton copula is chosen as \(\theta \approx 0.723\) so that their Kendall's tau value  \(\tau\approx 0.266\) is the same as for the specified Gaussian copula. The copulas model the dependencies between the default indicator \(X_n\) and the latent factor variable \(Y\) in the respective threshold model. 
    }
    \label{fig:sampleCopulaModels}
\end{figure}
    
    \item \label{exBMMb} The Clayton copula family \((C_{\theta}^{Cl})_{\theta\in [-1,\infty)}\) belongs to the class of Archimedean copulas and is given by \(
    C_\theta^{Cl}(u,v) = \max\left\{0,\left(u^{-\theta}+v^{-\theta}-1\right)^{-1 / \theta}\right\}\)
for \(\theta\ne 0\) and by \(C_\theta^{Cl}(u,v) = uv\) for \(\theta =0,\)
see \cite[Chapter 4]{Nelsen-2006} or 
\cite[Example 11.13]{Embrechts-2015}.
For \(\theta >0,\) it holds that
\(C_{\theta}^{Cl}(u|v) = v^{-\theta-1} (u^{-\theta} + v^{-\theta} -1)^{-1/\theta -1},\)
which decreases in \(v,\) so that \(C_{\theta}^{Cl}\) is SI. 
Taking the inverse of \(C_{\theta}^{Cl}(u|v)\) with respect to \(u,\) we obtain from Proposition \ref{lemconBMM} for the Clayton copula model \(F_{X_n,Y} = C_{\theta}^{Cl}\) the closed-form expression
    \begin{align}\label{Claytoncopmod}
        X_n = f_n(Y,\varepsilon_n) 
        = (C_{\theta}^{Cl})^{-1}(\varepsilon_n|Y) 
        = \left[(\varepsilon_n Y^{\theta+1})^{-\theta/(\theta+1)} - Y^{-\theta} +1 \right]^{-1/\theta}.
    \end{align}
For \(\theta > 0\) the Clayton copula admits lower (but no upper) tail dependence; see e.g. \cite[Table 3]{Ansari-Rockel-2023}. 
Hence, the Clayton copula threshold model may be used for modeling default indicator variables where the underlying random variable \(X_n\) depends for small values stronger on the risk factor \(Y\) than in the Gaussian copula setting. 
From \eqref{survcop}, we obtain 
that  \(C_{D_n,Z} = \widehat{C_\theta^{Cl}}.\)
\item For comparing the impact of tail dependencies on the total credit portfolio loss, we also consider \emph{survival Clayton copulas} given by
\begin{align}
    C_{\theta}^{sCl}(u,v) := \widehat{C_{\theta}^{Cl}}(u,v) = u + v - 1 + C_{\theta}^{Cl}(1-u,1-v)
\end{align}
for \((u,v)\in [0,1]^2.\) Due to straightforward symmetry arguments, we obtain for \(\theta>0\) that \(C_{\theta}^{sCl}\) is SI and exhibits upper but no lower tail dependence.
The copulas of the siBMM, implied by the survival Clayton copula threshold model, are given by \(C_{D_n,Z} = \widehat{C_\theta^{sCl}} = C_\theta^{Cl}.\)
    \end{enumerate}
\end{example}



\subsection{Comparison of Threshold Models}\label{secsompdefmod}

For the proof of Theorem \ref{propcomBMMs} and for the robustness results for siBMMs in terms of default integral functions, we use simple conditions on the bivariate SI copulas \(C_n = C_{X_n,Y},\) \(n\in \{1\,\ldots,N\},\) that enable the comparison of losses in siBMMs relying on the copula-based threshold models in Proposition \ref{lemconBMM}.

The following proposition states that a pointwise comparison of SI copulas in threshold models yields a comparison of credit portfolio losses in convex order. 

\begin{proposition}[Comparison of threshold models in convex order]\label{lemconord}~\\
    For all \(n\in \{1,\ldots,N\},\) let \(X_n,X_n',Y,Y'\sim U(0,1)\) be random variables with SI copulas \(C_{X_n,Y}\) and \(C_{X_n',Y'}.\) Assume that \(X_1,\ldots,X_N\) are conditionally independent given \(Y\) and that \(X_1',\ldots,X_N'\) are conditionally independent given \(Y'.\) Consider the threshold models \(D_n = \1_{\{X_n \leq c_n\}}\) and \(D_n'=\1_{\{X_n'\leq c_n\}}.\) Let \(l_1,\ldots,l_N\) be independent non-negative random variables that are also independent from \(\{D_n\}_n\) and \(\{D_n'\}_n.\) If \(C_{X_n,Y}(u,v)\leq C_{X_n',Y'}(u,v)\) for all \((u,v)\in [0,1]^2\) and \(n\in \{1,\ldots,N\},\) then 
\begin{align}\label{eqlemconord}
    \sum_{n=1}^N l_n D_n \leq_{cx} \sum_{n=1}^N l_n D_n' .
\end{align}
\end{proposition}


\begin{example}\label{excomthr}
Under the assumptions of Proposition \ref{lemconord}, assume that, for all \(n,\) \((X_n,Y)\) and \((X_n',Y')\) follow a Gaussian / Clayton / survival Clayton copula model, i.e., \(C_{X_n,Y} = C_{\rho_n}^{Ga}\) and \(C_{X_n',Y'} = C_{\rho_n'}^{Ga}\) for \(\rho_n,\rho_n'\geq 0\) or \(C_{X_n,Y} = C_{\theta_n}^{Cl}\) and \(C_{X_n',Y'} = C_{\theta_n'}^{Cl}\) or \(C_{X_n,Y} = C_{\theta_n}^{sCl}\) and \(C_{X_n',Y'} = C_{\theta_n'}^{sCl}\) for \(\theta_n,\theta_n'\geq 0.\) Since the family of Gaussian / Clayton / survival Clayton copulas is increasing in its parameter w.r.t. the pointwise order and the copulas are SI for non-negative parameters, we obtain that \(\rho_n\leq \rho_n'\), resp. \(\theta_n\leq \theta_n'\),  for all \(n\) implies \eqref{eqlemconord}. 
    Further well known copula families that fulfill the assumptions of Proposition \ref{lemconord} are given in \cite{Ansari-Rockel-2023}. 
\end{example}

\begin{remark}
    If the distribution function of the common risk factor \(Z\) is continuous\footnote{If \(F_Z\) is discontinuous, then the conditional default probability \(\mathbb{P}(D_n=1 \mid Z=F_Z^{-1}(s))\) can be similarly expressed as a generalized partial derivative of the underlying copula \(C,\) see \cite[Theorem 2.2]{Ansari-Rueschendorf-2021}. Then the representation of \(G_{D_n,Z}\) in \eqref{Gnexp} has to be modified slightly.}, the default integral function \(G_{D_n,Z}\) has a simple representation through the copula of \((D_n,Z)\) or \((X_n,Y)\) as
\begin{align}\label{Gnexp}\begin{split}
   G_{D_n,Z}(t) &= \int_0^t \mathbb{P}(D_n=1 \mid Z=F_Z^{-1}(s)) \de s
    = \int_0^t [1 - \mathbb{P}(D_n=0 \mid Z = F_Z^{-1}(s))] \de s\\
   & = t - \int_0^t \partial_2 C_{D_n,Z}(1-\pi_n, F_Z(F_Z^{-1}(s))) \de s
    = t - C_{D_n,Z}(1-\pi_n,t) \\
    &= t - \hat{C}_{X_n,Y}(1-\pi_n,t) = \pi_n - C_{X_n,Y}(\pi_n,1-t),
    \end{split}
\end{align}
where the first equality is due to the definition in \eqref{defcdpf}. The third equality follows from the representation of the conditional distribution function in \eqref{repcondcdf} and the transformation \(V=F_Z(Z)\) using continuity of \(F_Z.\) The fourth equality follows from \(F_Z(F_Z^{-1}(s)) = s\) for all \(s\) using again continuity of \(F_Z.\) The last two equalities are due to \eqref{survcop}.
Hence, by \eqref{Gnexp}, a pointwise comparison of the default integral functions \(G_{D_n,Z}\) follows from a pointwise comparison of the underlying copulas in the siBMMs as in Proposition \ref{lemconord}. 
\end{remark}

\section{Proofs}\label{appendB}

This appendix addresses the proofs of all results in this paper. 
We move the proof of Theorem \ref{propcomBMMs} after the proof of Proposition \ref{lemconord} because it is based on several results from Appendix \ref{secmain}.

 \begin{proof}[Proof of Lemma \ref{lemBMM}:]
Condition \eqref{ass2} follows from the fact that functions of independent random variables are independent. For condition \ref{ass3}, we observe that
    \begin{align*}
        \mathbb{P}( D_n = 1 \mid Z=z ) &= \mathbb{P}(X_n\leq c_n \mid Z = z) = \mathbb{P}( f_n(Y,\varepsilon_n) \leq c_n \mid g(Y) = z)\\
    &= \mathbb{P}(f_n(g^{-1}(z),\varepsilon_n)\leq c_n \mid Y = g^{-1}(z))
    = \mathbb{P}(f_n(g^{-1}(z),\varepsilon_n) \leq c_n),
    \end{align*}
    which is increasing in \(z\) because \(f\) is increasing in its first component and \(g\) (and thus also \(g^{-1}\)) is strictly decreasing. The last equality follows from independence of \(\varepsilon_n\) and \(Y.\)\\
    Under the continuity assumptions, it can easily be verified that \(X_n\) has a continuous distribution function. Hence, \(c_n = F_{X_n}^{-1}(\pi_n)\) implies \(\pi_n = \mathbb{P}(X_n \leq c_n) = \mathbb{P}(D_n = 1).\) 
\end{proof}

\begin{proof}[Proof of Proposition \ref{lemconBMM}:] We first show \eqref{lemconBMM3}.
    Let \((x,y)\in [0,1].\) Then we have
    \begin{align*}
        F_{X_n,Y}(x,y) &=\mathbb{P}(X_n\leq x, Y\leq y) = \int_0^y \mathbb{P}(f_n(Y,\varepsilon_n)\leq x \mid Y=v)\de v \\
        &= \int_0^y \mathbb{P}(f_n(v,\varepsilon_n)\leq x )\de v = \int_0^y \mathbb{P}(C_n^{-1}(\varepsilon_n| v) \leq x) \de v
        = \int_0^y \mathbb{P}(\varepsilon_n \leq C_n(x|v))\de v \\
        &= \int_0^y C_n(x|v) \de v = \int_0^y \partial C_n(x,v)/\partial v \de v = C_n(x,y)=C_{X_n,Y}(x,y),
    \end{align*}
    where the third equality follows from independence of \(Y\) and \(\varepsilon_n.\) For the sixth equality, we use that \(\varepsilon_n\) is uniformly distributed on \((0,1).\) 
    \\
    Statement \eqref{lemconBMM1} follows from \eqref{lemconBMM3} because \(F_{X_n,Y}(x,1) = C_n(x,1) = x\) for all \(x\in [0,1].\)\\
    To show statement \eqref{lemconBMM2}, first observe that \(f_n(v,t)\) is increasing in its second component~$t$ because \(C_n(t|v)\) and thus \(C_n^{-1}(t|v)\) are increasing in \(t.\) Further, \(f_n(v,t) = C_n^{-1}(t|v)\) is increasing in $v$ for all~\(t\,\). This follows from the fact that $C_n$ is SI, which in turn implies the concavity of \(C_n(t,v)\) in $v$ for all \(t\), and the fact that \(C_n(t|v) = (\partial C_n(t,v)/\partial v)\) is decreasing in \(v\) for all \(t\) if and only if \(C_n^{-1}(t|v)\) is increasing in \(v\) for all \(t\) (note that the inverse is w.r.t. the first component $t$ while the distribution functions are compared for different $v$).  \\
    Since \(X_n\sim U(0,1),\) the default indicator \(D_n\) satisfies condition \eqref{ass1}. Properties \eqref{ass2} and \eqref{ass3} follow from Lemma \ref{lemBMM}.
    \end{proof}

\begin{proof}[Proof of Lemma \ref{correpBMM}:]
    We construct \(X_n\) and \(Y\) through the distributional transform of \(D_n\) and \(Z.\) To this end, define for a random variable \(\xi\) its distributional transform by
    \(\tau_\xi := F_\xi(\xi,W):= F_\xi(\xi-) + W (F_\xi(\xi)-F_\xi(\xi-)),\)
    where \(W\sim U(0,1)\) is independent from \(\xi\) and where \(F(x-)\) denotes the left-hand limit of \(F\) at \(x.\) It is
    \begin{align}\label{propdisttraf}
        \tau_\xi &\sim U(0,1) \quad \text{and} \quad \xi = F_\xi^{-1}(\tau_\xi) \quad P\text{-almost surely},
    \end{align}
    see \cite[Proposition 2.1]{Rueschendorf-2009}. For 
    \(X_n':= \tau_{D_n} = F_{D_n}(D_n,U_n)\) and \(Y':= \tau_Z = F_Z(Z,V),\) with \(U_1,\ldots,U_N,V\sim U(0,1)\) i.i.d. and independent from \(\{D_n\}_n\) and \(Z,\) define 
    \begin{align}\label{defXn}
        X_n:= 1-X_n' \quad \text{and} \quad Y:=1-Y'.
    \end{align}
    Then \eqref{propdisttraf} implies \(X_n,Y\sim U(0,1).\) Further, we have
    \begin{align}\label{correpBMM3}
        X_n &\leq c_n & &\Longleftrightarrow  & X_n'&\geq 1-c_n\\
        \nonumber &&&\Longleftrightarrow & F_{D_n}(D_n-) + U_n (F_{D_n}(D_n)-F_{D_n}(D_n-)) &\geq 1-c_n \\
        \nonumber &&&\Longleftrightarrow & D_n&> 0\\
        \nonumber &&&\Longleftrightarrow & D_n &=1,
    \end{align}
    where the third equivalence results from \(c_n = \pi_n\) (note that $c_n=F_{X_n}^{-1}(\pi_n)=\pi_n$ here since $X_n\sim U(0,1)$) and the definition of the distributional transform, which can be seen by a sketch of \(F_{D_n}.\) Hence, it holds that \(D_n = \1_{\{X_n\leq c_n\}}.\)   
    
    To show the SI property, we distinguish three cases and first assume that \(x < c_n.\) 
    Denote by \(\sigma(Z,V)\) the \(\sigma\)-algebra generated by \((Z,V).\)
    Then, for any event \(A\in \sigma(Z,V)\) we have
    \begin{align}\label{correpBMM8}
        \mathbb{P}(F_{D_n}(D_n,U_n) \geq 1-x \mid A) &= \int_0^1 \mathbb{P}(F_{D_n}(D_n,U_n) \geq 1-x \mid A, U_n = u) \de u\\
        \nonumber & = \int_0^1 \mathbb{P}(F_{D_n}(D_n,u) \geq 1-x \mid A) \de u = \frac{x}{c_n} \mathbb{P}(D_n = 1 \mid A),
    \end{align}
    where the second equality follows from independence of \(U_n\) from \((D_n,Z,V).\) For the third equality, we obtain for \(x< c_n\) that \(F_{D_n}(D_n,u)\geq 1-x\) is equivalent to \(D_n = 1\) and \(u\geq 1-x/c_n\,\) so that
$$
\int_0^1 \mathbb{P}(F_{D_n}(D_n,u) \geq 1-x \mid A) \de u =\frac{x}{c_n} \mathbb{P}(D_n = 1 \mid A)
$$    
as in the previous case. If $x=c_n$, the condition \(F_{D_n}(D_n,u)\geq 1-c_n=1-\pi_n\) is satisfied when $D_n=0$ only for $u=1$ and it holds for $D_n=1$ whenever \(u\geq 1-x/c_n=0\,\) which is always the case. Hence, in that case we have 
$$
\int_0^1 \mathbb{P}(F_{D_n}(D_n,u) \geq 1-x \mid A) \de u =\int_0^1  \mathbb{P}(D_n=1 \mid A) \de u  =\frac{x}{c_n} \mathbb{P}(D_n = 1 \mid A).
$$ 
    For \(x>c_n\) and hence $1-x<1-c_n$, we have for $D_n=1$ that \( F_{D_n}(D_n,u)=1-\pi_n+u\pi_n \geq 1-x\) for $u\geq 1-\frac{x}{\pi_n}$. Moreover, for $D_n=0$ we have \(F_{D_n}(D_n,u)=u (1-\pi_n)\geq 1-x\) for $u\geq \frac{1-x}{1-\pi_n}.$  Thus, we obtain (since $\pi_n=c_n$) that
    \begin{align}\label{correpBMM9}
        \mathbb{P}(F_{D_n}(D_n,U_n) \geq 1-x \mid A) &=  \int_{\frac{1-x}{1-\pi_n}}^1 \mathbb{P}(D_n=0 \mid A) \de u+ \int_{1-\frac{x}{\pi_n}}^1 \mathbb{P}(D_n=1 \mid A) \de u\\
        \nonumber  &= \left(1 - \frac{1-x}{1-c_n}\right) (1-\mathbb{P}(D_n = 1 \mid A))+ \frac{x}{c_n} \mathbb{P}(D_n=1\mid A) \\
        &= \frac{x-c_n}{1-c_n} + \left(\frac{x}{c_n} - \frac{x-c_n}{1-c_n}\right) \mathbb{P}(D_n = 1 \mid A).
    \end{align}
    We obtain for a random variable \(V\sim U(0,1)\) independent of \(\{D_n\}_n,\) \(\{U_n\}_n,\) and \(Z\) that
    \begin{align}\label{correpBMM1}
        \mathbb{P}(X_n\leq x \mid Y=y) 
        &= \mathbb{P}(X_n'\geq 1-x \mid \tau_Z = 1-y)\\
        \nonumber&= P\left(F_{D_n}(D_n,U_n)\geq 1-x\mid F_{Z}(Z,V) = 1-y\right)\\
        \nonumber&= \begin{cases}
            \frac{x}{c_n} \mathbb{P}(D_n = 1 \mid Z = F_Z^{-1}(1-y), V = v_y), &\text{if } x\leq c_n,\\
           \frac{x-c_n}{1-c_n} + \left(\frac{x}{c_n} - \frac{x-c_n}{1-c_n}\right) \mathbb{P}(D_n=1\mid Z = F_Z^{-1}(1-y), V = v_y), &\text{if } x > c_n,\\
        \end{cases} \\
        \nonumber&= \begin{cases}
            \frac{x}{c_n} \mathbb{P}(D_n = 1 \mid Z = F_Z^{-1}(1-y)), &\text{if } x\leq c_n,\\
           \frac{x-c_n}{1-c_n} + \left(\frac{x}{c_n} - \frac{x-c_n}{1-c_n}\right) \mathbb{P}(D_n=1\mid Z = F_Z^{-1}(1-y)), &\text{if } x > c_n,\\
        \end{cases} 
    \end{align}
    for some \(v_y\in \R\) depending only on \(F_Z\) and \(y\) such that \(F_{Z}(Z-) + V(F_{Z}(Z)-F_{Z}(Z-)) = 1-y\) for \(Z = F_Z^{-1}(1-y).\)
    For the third equality we also apply \eqref{correpBMM8} and \eqref{correpBMM9}.
    The last equality follows from independence of \(V\) from \(D_n\) and \(Z.\) Since by condition \eqref{ass3} the conditional default probability is increasing in \(z,\) we obtain from \eqref{correpBMM1} that \(\mathbb{P}(X_n\leq x \mid Y=y)\) is decreasing in \(y.\) Hence, \(C_{X_n,Y}(u,v) = \int_0^v \mathbb{P}(X_n\leq u \mid Y=y) \de y\) is concave in \(v\) for all \(u,\) which implies that \(C_{X_n,Y}\) is SI.

    It remains to prove conditional independence of \(X_1,\ldots,X_N\) given \(Y,\) which follows from
    \begin{align*}
        & P\big(\bigcap \{X_n\leq x_n\} \mid Y=y\big) = P\big(\bigcap \{X_n'\geq 1-x_n\} \mid Y' = 1-y\big) \\
        \nonumber &= P\big(\bigcap \{F_{D_n}(D_n,U_n)\geq 1-x_n\} \mid F_{Z}(Z,V) = 1-y\big) \\
        \nonumber &= \int_{[0,1]^{N+1}} P\big(\bigcap \{F_{D_n}(D_n,u_n)\geq 1-x_n\} \mid F_{Z}(Z,v) = 1- y\big) \de \lambda^{N+1}(u_1,\ldots,u_N,v)\\
        \nonumber &= \int_{[0,1]^{N+1}} P\big(\bigcap \{D_n\geq F_{D_n}^{-1}(1-x_n)\} \mid Z = F_Z^{-1}(1- y)\big) \de \lambda^{N+1}(u_1,\ldots,u_N,v)\\
        \nonumber &= \int_{[0,1]^{N+1}} \prod_n P\big(D_n\geq F_{D_n}^{-1}(1-x_n) \mid Z = F_Z^{-1}(1- y)\big) \de \lambda^{N+1}(u_1,\ldots,u_N,v)\\
        \nonumber &= \prod_n P\big(D_n\geq F_{D_n}^{-1}(1-x_n) \mid Z = F_Z^{-1}(1- y)\big) \\
        \nonumber &=  \prod_n \int_{[0,1]^2} P\big(D_n\geq F_{D_n}^{-1}(1-x_n) \mid Z = F_Z^{-1}(1- y)\big) \de \lambda^2 (u,v)\\
        \nonumber &= \cdots = \prod_n \mathbb{P}(X_n\leq x_n\mid Y=y),
    \end{align*}
    where we use for the third equality that \(U_1,\ldots,U_N,V\) are i.i.d. uniformly distributed on \((0,1),\) for the fourth equality, we use independence of \(U_1,\ldots,U_N,V\) from \(\{D_n\}_n\) and \(Z.\) For the fifth equality, we use conditional independence of \(D_1,\ldots,D_N\) given \(Z.\)\\
    The property that \(C_{D_n,Z} = \hat{C}_{X_n,Y}\) (which is SI) follows from Remark \ref{remSIC}\eqref{remSIC2}.
\end{proof}

\begin{proof}[Proof of Proposition \ref{lemconord}:]
By assumption, the random variables \(X_1,\ldots,X_N\) are conditionally independent given \(Y\) and the random variables \(X_1',\ldots,X_N'\) are conditionally independent given \(Y'.\) Hence, we obtain from \cite[Corollary 4]{Ansari-Rueschendorf-2023} that \((X_1,\ldots,X_N) \leq_{sm} (X_1',\ldots,X_N'),\) using that \(X_n \eqd X_n'\) for all \(n\) and \(Y\eqd Y'\) and using that the copulas \(C_{X_n,Y}\) and \(C_{X_n',Y'}\) are SI and satisfy \(C_{X_n,Y}(u,v) \leq C_{X_n',Y'}(u,v)\) for all \((u,v)\in [0,1]^2,\) where \(\leq_{sm}\) denotes the supermodular order which compares expectations of supermodular functions of random vectors.
    Since the supermodular order is closed under decreasing transformations (\cite[Theorem 9.A.9(a)]{Shaked-Shantikumar-2007}), it follows that \((D_1,\ldots,D_N)\leq_{sm} (D_1',\ldots,D_N').\) Applying \cite[Corollary 9.A.10]{Shaked-Shantikumar-2007}, we obtain \((l_1 D_1,\ldots,l_N D_N)\leq_{sm} (l_1 D_1',\ldots, l_N D_N')\) using non-negativity as well as independence of \(\{l_n\}_n\) from \(\{D_n\}_n\) and \(\{D_n'\}_n.\) Since for any convex function \(\varphi\colon \R\to \R,\) the function \((x_1,\ldots,x_N) \mapsto \varphi(x_1+\cdots + x_N)\) is a supermodular function, see \cite[Theorem 8.3.3 and its proof]{Mueller-Stoyan-2002}, the convex comparison in \eqref{eqlemconord} follows.
\end{proof}

\begin{proof}[Proof of Theorem \ref{propcomBMMs}:]
We construct copula-based threshold models which coincide with the siBMMs and are built by SI copulas that allow a pointwise comparison. 
Then, Proposition \ref{lemconord} implies the statement.

To this end, consider \(\{X_n\}_n\) and \(Y\) defined in \eqref{defXn}. We obtain from Lemma \ref{correpBMM} that \(X_n,Y\sim U(0,1)\) and \(C_{X_n,Y}\) is SI for all \(n.\) Further, \(X_1,\ldots,X_N\) are conditionally independent given \(Y\) and we have \(D_n = \1_{\{X_n \leq c_n\}}\) for all \(n.\)
Similarly, consider random variables \((X_n')_n\) and \(Y'\) associated with the default model \((D_1',\ldots,D_N',Z').\)
Then, we obtain for all \((u,v)\in [0,1]^2\) and for all \(n\) that
\begin{align*}
    C_{X_n,Y}(u,v) &= \int_0^v \mathbb{P}(X_n\leq u \mid Y = y) \de y  = u - \int_v^1 \mathbb{P}(X_n\leq u \mid Y = y) \de y \\
    & = \begin{cases}
        u - \int_v^1 \frac{u}{c_n} \mathbb{P}(D_n = 1 \mid Z = F_Z^{-1}(1-y)) \de y, &\text{if } u \leq c_n,\\
        u - \int_v^1 \left[\frac{u-c_n}{1-c_n} + \left(\frac{u}{c_n} - \frac{u-c_n}{1-c_n}\right) \mathbb{P}(D_n=1\mid Z = F_Z^{-1}(1-y))\right] \de y, &\text{if } u > c_n,
    \end{cases}\\
    & = \begin{cases}
        u - \frac{u}{c_n} \int_0^{1-v}  \mathbb{P}(D_n = 1 \mid Z = F_Z^{-1}(y')) \de y', &\text{if } u \leq c_n,\\
        u-\frac{(1-v)(u-c_n)}{1-c_n} -\left(\frac{u}{c_n}-\frac{u-c_n}{1-c_n}\right) \int_0^{1-v} \mathbb{P}(D_n = 1 \mid Z = F_Z^{-1}(y')) \de y' , &\text{if } u > c_n,
    \end{cases}\\
    & = \begin{cases}
        u - \frac{u}{c_n} G_{n,Z}(1-v), &\text{if } u \leq c_n,\\
        u-\frac{(1-v)(u-c_n)}{1-c_n} -\left(\frac{u}{c_n}-\frac{u-c_n}{1-c_n}\right) G_{D_n,Z}(1-v) , &\text{if } u > c_n,
    \end{cases}\\
    & \leq \begin{cases}
        u - \frac{u}{c_n} G_{n,Z'}(1-v), &\text{if } u \leq c_n,\\
        u-\frac{(1-v)(u-c_n)}{1-c_n} -\left(\frac{u}{c_n}-\frac{u-c_n}{1-c_n}\right) G_{D_n,Z'}(1-v) , &\text{if } u > c_n,
    \end{cases}\\
    & = \cdots = C_{X_n',Y'}(u,v),
\end{align*}
where the third equality follows from \eqref{correpBMM1}. For the fourth equality, we transform \(y':= 1-y.\) The fifth equality holds true by definition of \(G_{n,Z}\) in \eqref{defcdpf}, and the inequality is due to the assumption that \(G_{n,Z}(z) \geq G_{n,Z^*}(z)\) for all \(z,\) where we use that \(\frac{u}{c_n}-\frac{u-c_n}{1-c_n} \geq 0.\)
\end{proof}

\begin{proof}[Proof of Theorem \ref{thebmmsiBMM}:]
    Since \(p_{D_n'} \circ F_{Z'}^{-1}\) is the increasing rearrangement of \(p_{D_n} \circ F_{Z}^{-1},\) it follows that 
    \begin{align}\label{defSchfun}
        F_{D_n|Z=F_Z^{-1}(\cdot)}(y) =_S F_{D_n'|Z'=F_{Z'}^{-1}(\cdot)}(y) \quad \text{for all } y \in \R,
    \end{align}
    where \(=_S\) denotes equality in the Schur order for functions, i.e., the increasing rearrangements of the left and right sides of \eqref{defSchfun} coincide. Due to the assumptions, we know that \(D_n'\) is stochastically increasing in \(Z',\) and \(D_n\eqd D_n'\) for all \(n\) as well as \(Z\eqd Z'.\) Hence, we obtain from \cite[Proposition 4.2]{Ansari-Ritter-2024} that \((D_1,\ldots,D_N)\leq_{sm} (D_1',\ldots,D_N').\) Then, proceeding in the same way as in the end of the proof of Proposition \ref{lemconord}, we obtain \eqref{eqthebmmsiBMM}.
\end{proof}

\begin{proof}[Proof of Theorem \ref{mainthe}:]
Using \(G_{D_{n}^x,Z^x}\in \cF_{icx}^{\pi_n}\) for all \(x\in I\), it follows that \(\underline{G}_{n},\overline{G}_{n}\in \cF_{icx}^{\pi_n}\) with \(\underline{G}_n\geq G_{D_n^x,Z^x} \geq \overline{G}_n\) pointwise for all \(x.\) Then, the assertions follow from Theorem \ref{propcomBMMs} using the one-to-one correspondence between default integral functions and the class \(\cF_{icx}^{\pi_n}.\)
\end{proof}

\end{document}